\documentclass[10pt, two column, twoside]{IEEEtran}
\usepackage{amssymb}
\usepackage{amsmath} 
\usepackage[lined,boxed,commentsnumbered, ruled]{algorithm2e}
\usepackage{makecell}
\usepackage{mathrsfs}
\usepackage{bm}
\usepackage{tikz}
\usepackage{caption}
\usepackage{graphicx,booktabs,multirow}
\usepackage{epstopdf}
\usepackage{subfigure}
\usepackage{multirow}
\usepackage{tabularx} 
\usepackage{booktabs}
\usepackage{float}
\usepackage{verbatim}
\usepackage{xcolor}
\usepackage{color}
\usepackage{colortbl}
\usepackage[colorlinks, linkcolor=black, anchorcolor=black, citecolor=black]{hyperref}
\usepackage{cite}
\usepackage{setspace}
\usepackage{array}

\definecolor{colorhkust}{RGB}{20,43,140}
\definecolor{colortsinghua}{RGB}{116,52,129}
\definecolor{color1}{RGB}{128,0,0}
\usepackage{amsthm}

\usepackage[T1]{fontenc}
\usepackage[cmintegrals]{newtxmath}
\usetikzlibrary{arrows}
\theoremstyle{definition} 

\newtheorem{theorem}{Theorem}

\newtheorem{definition}{Definition}

\newtheorem{assumption}{Assumption}

\newcommand{\trace}{{\rm Tr}}

\pdfoutput=1

\begin{document}

\title{Differentially Private Federated Learning via Reconfigurable Intelligent Surface}
\author{Yuhan~Yang, \textit{Student Member}, \textit{IEEE}, Yong Zhou, \textit{Member}, \textit{IEEE}, Youlong Wu \textit{Member}, \textit{IEEE}, \\
and Yuanming~Shi, \textit{Senior Member}, \textit{IEEE}
\thanks{Y. Yang, Y. Zhou, Y. Wu, and Y. Shi are with ShanghaiTech University, Shanghai 201210, China (e-mail: \{yangyh1, zhouyong, wuyl1, shiym\}@shanghaitech.edu.cn).}   
}


\maketitle

\begin{abstract}
Federated learning (FL), as a disruptive machine learning paradigm, enables the collaborative training of a global model over decentralized local datasets without sharing them. It spans a wide scope of applications from Internet-of-Things (IoT) to biomedical engineering and drug discovery. To support low-latency and high-privacy FL over wireless networks, in this paper, we propose a reconfigurable intelligent surface (RIS) empowered over-the-air FL system to alleviate the dilemma between learning accuracy and privacy. This is achieved by simultaneously exploiting the channel propagation reconfigurability with RIS for boosting the receive signal power, as well as waveform superposition property with over-the-air computation (AirComp) for fast model aggregation. By considering a practical scenario where high-dimensional local model updates are transmitted across multiple communication blocks, we characterize the convergence behaviors of the differentially private federated optimization algorithm. We further formulate a system optimization problem to optimize the learning accuracy while satisfying privacy and power constraints via the joint design of transmit power, artificial noise, and phase shifts at RIS, for which a two-step alternating minimization framework is developed. Simulation results validate our systematic, theoretical, and algorithmic achievements and demonstrate that RIS can achieve a better trade-off between privacy and accuracy for over-the-air FL systems.
\end{abstract}


\IEEEpeerreviewmaketitle

\section{Introduction}\label{I}
With the rapid advancement of communication technologies for Internet-of-Things (IoT), massive amounts of sensory data generated by various edge devices (e.g., smartphones, wearables) can be leveraged to support various intelligent applications and services \cite{Intorduction-intelligent-networks}. However, the concern on data privacy makes the data sharing among edge devices unappealing and hinders the exploration of the potential values for the decentralized datasets. Recently, federated learning (FL) \cite{GoogleFL} has been recognized as a disruptive machine learning (ML) paradigm that is capable of preserving data privacy without sharing private raw data. FL spans a wide scope of applications from 6G \cite{Intorduction-intelligent-networks,6G-introduction,FL-challenges2}, IoT \cite{Dinh_CST21,Ahmed_IoTJ21}, and keyboard prediction \cite{wang2021field}, to the recent applications in healthcare informatics \cite{rieke2020future,IoMT1,IoMT2}, medical imaging \cite{Yan_JBHI21}, drug discovery \cite{chen2020fl}, and diabetes mellitus \cite{warnat2021swarm}. In particular, for cross-device FL, many edge devices collaboratively train a common model under the orchestration of a central edge server by periodically exchanging their local model updates. Despite the promising benefits, FL still encounters many challenges including statistical heterogeneity, communication bottleneck, as well as privacy and security concerns \cite{FL-challenges1,FL-challenges2,Privacy-introduction}.

\par
The communication overhead incurred by the periodical model update exchange is a key performance-limiting factor of FL. To enable efficient model aggregation, over-the-air computation (AirComp) has been recently proposed as a promising aggregation scheme by integrating communication and computation \cite{AirComp}. Specifically, multiple edge devices can upload the local models concurrently by exploiting the waveform superposition property of a multiple-access channel \cite{Broadband,AirComp, AirComp3, Tao_TWC21}, thereby achieving high spectrum efficiency and low transmission latency. In particular, AirComp was leveraged in \cite{Broadband} to reduce latency for FL over broadband channels compared with traditional orthogonal multiple access schemes. The learning performance of over-the-air FL was improved in \cite{AirComp} by jointly optimizing the receive beamforming design and device selection based on AirComp. The authors in \cite{AirComp3} also proposed a gradient sparsification scheme for over-the-air stochastic gradient method to reduce the dimensionality of the exchanged updates.

\par
Despite the aforementioned advantages, the heterogeneity of wireless links inevitably degrades the learning performance due to the magnitude attenuation and misalignment of signals received at the edge server. To alleviate the detrimental effect of wireless fading channels, reconfigurable intelligent surface (RIS) has recently been introduced as a key enabling technology to support fast and reliable model aggregation for over-the-air FL systems by reconfiguring the propagation environment \cite{Introduction-RIS, Channel-coefficient, Channel-invariant}. Specifically, RIS is a planar metasurface equipped with an array of cost-effective passive reflecting elements, which are orchestrated by a software-enabled controller to adjust the phase shifts of the incident signals \cite{RIS-challenges, RIS-principle, Huang_TWC19, Survey-RIS}. This helps improve the power of received signals, thereby enhancing the communication performance between all edge devices and the edge server. The waveform superposition property of a wireless multiple-access channel based on RIS is thus able to adapt to the local model updates, thereby improving the learning performance for over-the-air FL systems via jointly optimizing phase shifts, receive beamforming and edge device selection \cite{Channel-coefficient,Channel-invariant}.

\par
Besides, as illustrated in \cite{Privacy-leakage1,Privacy-leakage2,FLPrivacy-leakage}, transmitting the model updates (e.g., local gradients) can still cause privacy leakage. To quantify the privacy disclosure,
\begin{table*}[t]\footnotesize	
	\renewcommand\arraystretch{1.2}	
	\vspace*{-0.6em}
	\centering 
	\begin{tabular}{|>{\color{black}}c ||>{\color{black}}c|>{\color{black}}c||>{\color{black}}c|}
		\hline
		\textbf{Notation} & \textbf{Description}& \textbf{Notation} & \textbf{Description}\\
		\hline
		$K$ & Number of edge devices & $N$ & Number of passive reflecting elements of RIS \\
		\hline
		$D_k$ & Size of local dataset in the $k$-th edge device & $T$ & Total number of learning rounds \\
		\hline
		$e$ & Channel coherence length & $I$ & Number of communication blocks in one learning round \\
		\hline
		$\pmb\theta$ & $d$-dimensional model parameter & $(\epsilon,\delta)$ & Privacy level and failure probability \\
		\hline
		$\pmb g_{k,t}$ & Local update of edge device $k$ in the $t$-th learning round &
		$\pmb n_{k,t}(i)$ & Artificial noise of edge device $k$ in the $i$-th block \\
		\hline
		$\alpha_{k,t}(i)$ & Transmit scalar of edge device $k$ in the $i$-th block & $\eta_t(i)$ & Uniform power scaling factor in the $i$-th block \\
		\hline
		$\pmb s_{k,t}(i)$ & Pre-processed signal of edge device $k$ in the $i$-th block & $\pmb x_{k,t}(i)$ & Transmit signal of edge device $k$ in the $i$-th block \\
		\hline
		$\pmb r_t(i)$ & Aggregated signal in the $i$-th block of learning round $t$ & $\pmb w_t(i)$ & Additive white Gaussian noise in the $i$-th block \\
		\hline
		$h_{k,t}^d(i)$ & Channel coefficient from edge device $k$ to the edge server & $\pmb h_{k,t}^r(i)$ & Channel coefficient from edge device $k$ to RIS \\
		\hline
		$\pmb m_t$ & Channel coefficient from RIS to the edge server & $h_{k,t}(i)$ & Composite channel response of the $k$-th edge device \\
		\hline
		$\pmb\Theta_t(i)$ & Phase shift matrix of RIS in the $i$-th block & $P_0$ & Maximum transmit power of all edge devices \\
		\hline
		$\sigma^2_{k,t}(i)$ & Power of artificial noise $\pmb n_{k,t}(i)$ & $N_0$ & Power of wireless channel noise \\
		\hline
	\end{tabular}
	\caption{Summary of notations in this paper.}  
	\label{table1}
\end{table*}
a rigorous mathematical framework called differential privacy (DP) \cite{Differential-privacy1, Differential-privacy2} has been proposed. Various privacy-preserving mechanisms have been developed by injecting random perturbations that obey specific distributions, e.g., Gaussian \cite{Differential-privacy1}, Laplacian \cite{Laplacian}, and Binomial \cite{Binomial}. In particular, the authors in \cite{FLPrivacy-free} developed an AirComp based FL scheme to achieve privacy protection for free. This procedure is accomplished by leveraging the inherent wireless channel noise to protect user privacy. The authors in \cite{Privacy-anonymous} proved the inherent anonymity of AirComp which hides each private local update in the crowd to ensure high privacy, thereby reducing the amount of artificial noise added to the local updates. A disturbance scheme with additional perturbations added to a partial set of edge devices to benefit the whole system was proposed in \cite{Parital-noise}. However, the learning accuracy of FL may significantly degenerate due to the reduction of signal-to-noise ratio (SNR) for privacy guarantee \cite{FLPrivacy-free}, which yields a trade-off between accuracy and privacy.

\par
To balance the privacy-accuracy trade-off, the existing works \cite{FLPrivacy-free,Privacy-anonymous,Parital-noise,Privacy-power-allocation1,Privacy-power-allocation2} only focus on designing power allocation schemes, without considering the reconfigurability of the wireless environment. In comparison, we propose an RIS-enabled over-the-air FL system to balance the trade-off between privacy and accuracy, which is achieved by simultaneously exploiting the channel propagation reconfigurability with RIS and waveform superposition property with AirComp. Besides, we consider a practical local model updates transmission scheme by distributing the high-dimensional model updates (e.g., deep neural network model) across multiple communication blocks in one learning round \cite{Delay}. However, most of the previous works often make the lightweight assumption on the learning models, which requires the dimension of model updates to be small enough to be transmitted within one communication block \cite{FLPrivacy-free,Privacy-anonymous,Parital-noise,Privacy-power-allocation1,Privacy-power-allocation2,DP-proof,AirComp4,Channel-invariant}. Through convergence analysis and system optimization, we reveal that the RIS-enabled FL system can enhance system SNR and boost the received signal power to establish high learning accuracy while satisfying the privacy requirements.

\par
The major contributions of this paper are summarized as follows:
\begin{itemize}
\setlength{\itemsep}{0pt}
\setlength{\parsep}{0pt}
\setlength{\parskip}{0pt}
\setlength{\topsep}{0pt}
\item From a systematic perspective, we develop an RIS-enabled FL system with privacy guarantees for fast and reliable model aggregation to alleviate the dilemma between learning accuracy and privacy by exploiting channel propagation reconfigurability and waveform superposition property.
\item From a theoretical perspective, for high-dimensional model updates transmission across multiple communication blocks in one learning round, we propose a privacy-preserving transmission scheme based on DP and analyze its convergence behavior.
\item From an algorithmic perspective, we propose a two-step alternating minimization framework to jointly optimize the transmit power, artificial noise, and phase shifts for system optimization, leading to an optimal power allocation scheme.
\item Simulation results validate our systematic, theoretical, and algorithmic achievements and demonstrate the advantages of deploying RIS for enhancing both privacy and accuracy performance in the learning process.
\end{itemize}
The remainder of this paper is organized as follows. The privacy-preserving RIS-enabled FL system is presented in Section \ref{II}. Section \ref{III} analyzes the training loss, privacy, and power constraints, yielding a nonconvex system optimization problem. In Section \ref{IV}, an alternating minimization framework is proposed for solving this nonconvex problem. Simulation results are elaborated in Section \ref{V} to demonstrate the advantages of the RIS-enabled FL system. Finally, Section \ref{VI} concludes this work.

\par
\emph{Notations}: Italic and boldface letters denote scalar and vector (matrix), respectively. $\mathbb{R}^{m\times n}$ and $\mathbb{C}^{m\times n}$ denote the real and complex domains with the space of $m\times n$, respectively. For a positive integer $i$, we let $[i]\!\triangleq\!\{1,\ldots,i\}$. The operators $(\cdot)^T,(\cdot)^H,\trace\,(\cdot),\text{rank}\,(\cdot),\mathbb{E}\,(\cdot)$, and $\text{diag}\,(\cdot)$ represent the transpose, Hermitian transpose, trace, rank, statistical expectation, and diagonal matrix, respectively. $\land$ denotes logical and operation. The operator $|\cdot|$ is the cardinality of a set or the absolute value of a scalar number, and $||\cdot||$ denotes the Euclidean norm.

\section{System Model}\label{II}
In this section, we first introduce an RIS-enabled FL system based on AirComp for fast and accurate model aggregation, followed by presenting a DP-based privacy-preserving pipeline for privacy concerns. The notations are listed in Table \ref{table1}.

\subsection{Distributed Federated Learning}
As shown in Fig. \ref{fig1}, we consider an RIS-enabled wireless FL system consisting of one single-antenna edge server, $K$ single-antenna edge devices indexed by set $\mathcal{K}\!=\!\{1,\ldots,K\}$, and an RIS equipped with $N$ passive reflecting elements. We assume that each edge device $k\!\in\!\mathcal{K}$ has its own local dataset $\mathcal{D}_k$ with $D_k\!=\!|\mathcal{D}_k|$ data samples. For a given $d$-dimensional model parameter $\pmb\theta\!\in\!\mathbb{R}^d$, the local loss function of edge device $k$ is defined as
\begin{equation}
F_k(\pmb{\theta}) = \frac{1}{D_k}\sum_{(\pmb x,y)\in\mathcal{D}_k}f(\pmb x,y;\pmb{\theta}), \label{F_local}
\end{equation}
where $f(\pmb x,y;\pmb{\theta})$ represents the sample-wise loss function quantifying the prediction error of model $\pmb\theta$ on training sample $\pmb x$ with respect to its true label $y$ in $\mathcal{D}_k$. We assume that all local datasets have the same size \cite{Broadband}, i.e., $D_k=D,\;\forall k\in\mathcal{K}$. Then, the global loss function can be defined as
\begin{equation}
F(\pmb{\theta}) = \frac{1}{\sum_{k=1}^{K}D_k}\sum_{k = 1}^{K}D_kF_k(\pmb{\theta})=\frac{1}{K}\sum_{k=1}^{K}F_k(\pmb\theta) \label{F_global},
\end{equation}
which refers to the empirical average of the sample-wise loss functions on the global dataset $\mathcal{D}=\bigcup_{k=1}^{K}\mathcal{D}_k$. We also assume that each edge device observes i.i.d. samples from a common distribution \cite{FLPrivacy-free}. The non-i.i.d. setting in FL is also of special interest \cite{FLnoniid1}, but it is not the focus of this work.

\par
The learning process minimizes (\ref{F_global}) by updating $\pmb{\theta}$ based on the local gradients sent by edge devices, i.e.,
\begin{equation}
\pmb\theta^*=\underset{\pmb\theta\in\mathbb{R}^d}{\text{argmin}}\,F(\pmb\theta).
\end{equation}  
This problem can be tackled by using the popular decentralized optimization method, e.g., \pmb{FedSGD} \cite{GoogleFL}. For analytical ease, in this paper, full-batch gradient descent is adopted for local model update and $\pmb\theta$ is updated periodically within $T$ training rounds \cite{AirComp,FLPrivacy-free,Privacy-power-allocation2}. The $t$-th learning round consists of the following procedures:
\begin{itemize}
    \item \emph{Broadcast}: The edge server broadcasts the current global model parameter $\pmb\theta_{t}$ to all edge devices.
    \item \emph{Local update}: Each edge device $k$ computes its local update $\pmb g_{k,t}$ with respect to its local dataset $\mathcal{D}_k$ based on $\pmb\theta_t$, i.e.,
    \begin{equation}
    \pmb g_{k,t}=\nabla F_k\left(\pmb\theta_t\right)\in\mathbb{R}^d. \label{Local-Gradient}
    \end{equation} 
    \item \emph{Model aggregation}: All edge devices in $\mathcal{K}$ upload $\{\pmb g_{k,t}\}$ to the edge server. Then the edge server updates $\pmb\theta_{t+1}$ based on the aggregated local updates via gradient descent, i.e.,
    \begin{equation}
    \pmb\theta_{t+1}=\pmb\theta_{t}-\lambda\hat{\pmb g}_{t}, \label{Gradient-descent}
    \end{equation}
    where $\lambda$ denotes the learning rate and $\hat{\pmb g}_{t}$ is an estimation of the global gradient $\pmb g_t=\nabla F\left(\pmb\theta_{t}\right)=\frac{1}{K}\sum_{k=1}^{K}\pmb g_{k,t}$.
\end{itemize}

\begin{figure}[t]
        \centering
        \includegraphics[scale=0.40]{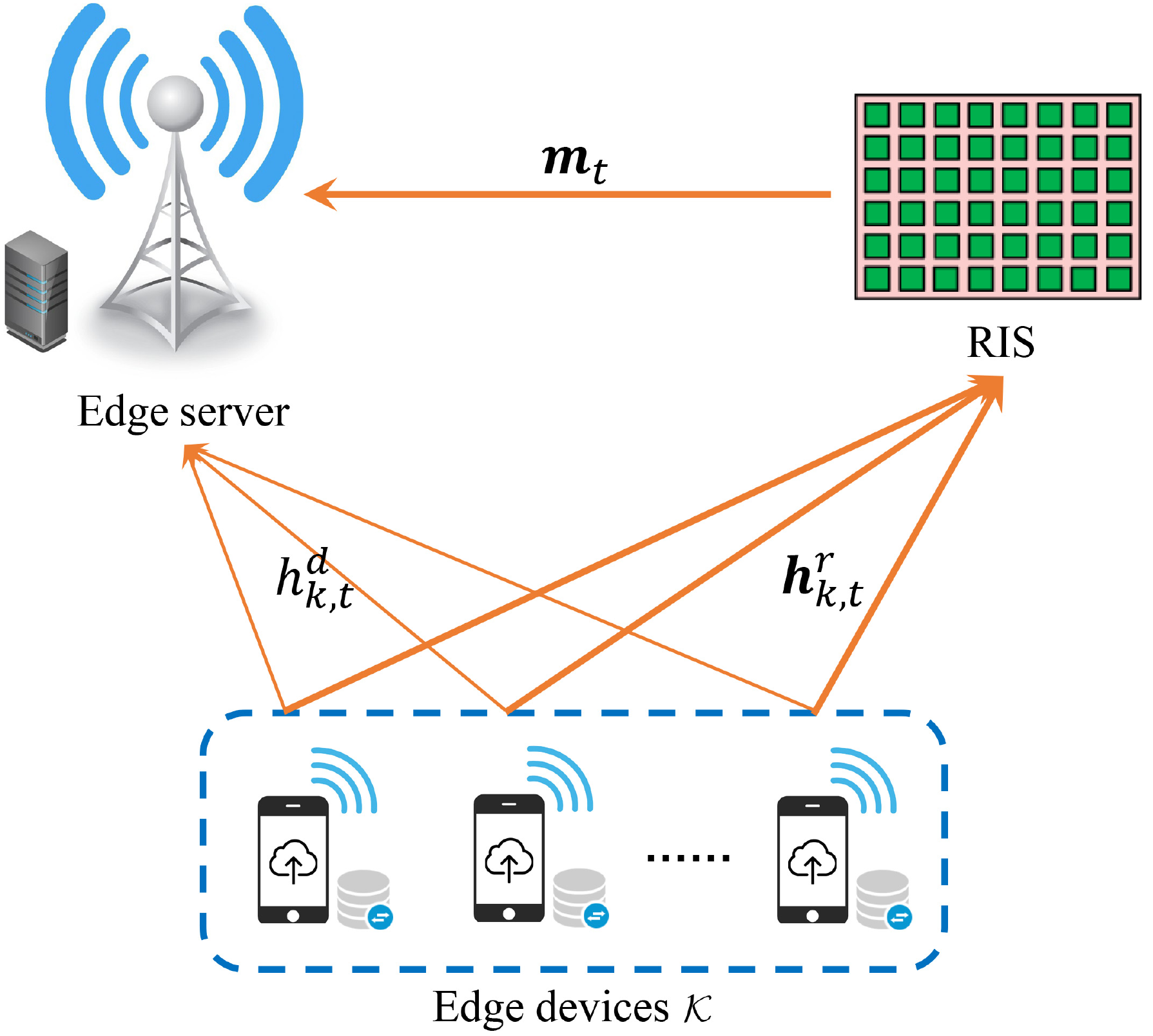}
        \caption{An RIS-enabled wireless FL system.}
        \label{fig1}
\end{figure}

\subsection{Channel Model for RIS-Enabled FL Systems}
The existence of unfavorable channel propagations and the power limitation of each edge device may severely degrade the accuracy for model aggregation, thereby reducing the learning accuracy of the FL system. To address this issue, an RIS equipped with $N$ reflecting elements is deployed to enhance the channel conditions of the channel links between all edge devices and the edge server \cite{Channel-invariant}.

\par
We focus on the information exchange process among the edge server, RIS, and edge devices, as shown in Fig. \ref{fig1}. To simplify the theoretical analysis, the downlink channels are assumed to be noise-free \cite{Phase-correction,AirComp4}, while we mainly consider the uplink fading channels for aggregating the local models \cite{Channel-coefficient}. As the dimension $d$ of the model parameters is often much larger than the channel coherence length $e$ \cite{Delay}, we propose a practical transmission scheme illustrated in Fig. \ref{fig2}, where each edge device uploads its local update sequentially over several consecutive communication blocks. Specifically, in the $t$-th learning round, the whole update message $\pmb x_{k,t}\!=\![\pmb x_{k,t}^T(1),\ldots,\pmb x_{k,t}^T(I)]^T\!\in\!\mathbb{C}^d$ (which will be elaborated in Section \ref{II}-D) that is a function of $\pmb g_{k,t}$ is evenly divided into $I\!=\!\lceil d/e\rceil$ $e$-dimensional vectors denoted by $\{\pmb x_{k,t}(i)\!\in\!\mathbb{C}^e\}$. Then $\{\pmb x_{k,t}(i)\}$ are sequentially transmitted across $I$ communication blocks. We further assume $I=d/e$ for simplicity.
\vskip -10pt
\begin{figure}[htb]
        \centering
        \includegraphics[scale=0.40]{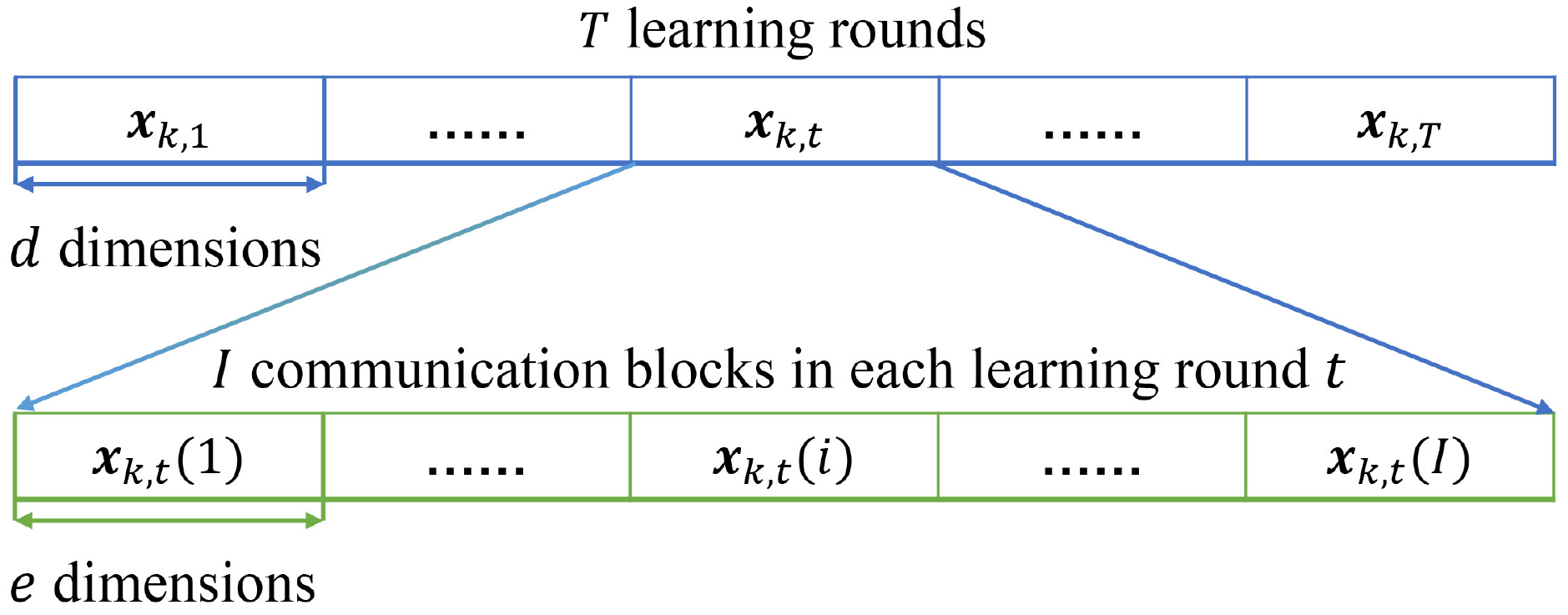}
        \caption{The transmission process of edge device $k$, where one learning round contains $I$ communication blocks.}
        \label{fig2}
\end{figure}
\vskip 5pt

\par
For the $i$-th communication block in learning round $t$, let $h_{k,t}^d(i)\in\mathbb{C}$, $\pmb{h}_{k,t}^r(i)\in\mathbb{C}^N$, and $\pmb{m}_t(i)\in\mathbb{C}^{N}$ denote the channel responses from edge device $k$ to the edge server, from edge device $k$ to RIS, and from RIS to the edge server, respectively. We assume that the channel coefficients are independent across different communication blocks and remain invariant during one communication block. We also assume that perfect channel state information (CSI) is available \cite{Channel-coefficient,Channel-invariant,Phase-correction}. Moreover, the phase-shift matrix of the RIS is denoted as $\pmb{\Theta}_t(i)=\text{diag}(\beta e^{j\phi_{1,t}(i)},\ldots,\beta e^{j\phi_{N,t}(i)})\in\mathbb{C}^{N\times N}$, where $\phi_{n,t}(i)\in[0,2\pi),\,n\in\{1,\ldots,N\}$. Without loss of generality, the amplitude reflection coefficient $\beta$ at each reflecting element is set to be one \cite{RIS-amplitude-1}. We assume that $\pmb\Theta_t(i)$ is designed at the edge server and then transmitted to the RIS via error-free downlink channels. We also assume that signals reflected twice or more times by the RIS can be ignored \cite{RIS-principle}. The composite channel response of the $k$-th edge device is thus denoted as $h_{k,t}(i)=\pmb m_t^T(i)\pmb\Theta_t(i)\pmb h_{k,t}^r(i)+h_{k,t}^d(i)$.

\par
In the $i$-th communication block of learning round $t$, the aggregated signal at the edge server is given by
\begin{equation}
\pmb r_{t}(i)=\sum_{k=1}^{K}h_{k,t}(i)\,\pmb x_{k,t}(i)+\pmb w_t(i), \label{Received-signal}
\end{equation}
where $\pmb w_{t}(i)\!\sim\!\mathcal{CN}(\pmb 0,N_0\pmb I_e)$ is the additive white Gaussian noise (AWGN) and the aggregated signal $\pmb r_t\!\in\!\mathbb{C}^d$ can be written as $\pmb r_t\!=\![\pmb r_t^T(1),\ldots,\pmb r_t^T(I)]^T$ for the $t$-th learning round.

\subsection{Differential Privacy}
In the RIS-enabled FL system, the sensitive raw data never leaves the edge devices to protect users' privacy. Nevertheless, as revealed in \cite{FLPrivacy-leakage}, the transmission of model updates $\pmb g_{k,t}$ may leak information about the local dataset statistically, which requires additional mechanisms to provide stronger privacy guarantees. We assume that the edge server is honest-but-curious \cite{DP-proof}, i.e., the edge server aims to infer the local information of each edge device based on the received signals $\pmb r=\{\pmb r_t\}_{t=1}^{T}$. We introduce DP based on the concept of neighboring datasets. Let $\mathcal{D}_k\!=\!\{\pmb x_1,\ldots,\pmb x_n\}\!\in\!\mathcal{X}^n$ denote a dataset comprising $n$ data points from $\mathcal{X}$. Two datasets $\mathcal{D}_k=\{\pmb x_1,\ldots,\pmb x_n\}$ and $\mathcal{D}_k'=\{\pmb x_1',\ldots,\pmb x_n'\}$ with the same cardinality are neighboring if they differ only by one element, i.e., there exists an index $i\in[n]$ such that $\pmb x_i\neq\pmb x_i'$ and $\pmb x_j=\pmb x_j'$ for all $i\neq j$. We then present the following definition \cite{Differential-privacy1, Differential-privacy2} for DP.

\begin{definition}
Given $\epsilon>0$, $0\leq\delta\leq 1$, and an arbitrary protocol $\mathcal{M}: \mathcal{X}^n\rightarrow\mathcal{Y}$. For any two possible neighboring datasets $\mathcal{D}_k, \mathcal{D}_k'\in\mathcal{X}^n$ and any subset $\Lambda\subseteq\mathcal{Y}$, protocol $\mathcal{M}$ is $(\epsilon,\delta)$-differentially private (in short, $(\epsilon,\delta)$-DP), if the following inequality holds:
\begin{equation}
\Pr[\mathcal{M}(\mathcal{D}_k)\in\Lambda]\leq \text{e}^{\epsilon}\Pr[\mathcal{M}(\mathcal{D}_k')\in\Lambda]+\delta. \label{DP}
\end{equation} 
Note that the case of $\delta=0$ is called pure $\epsilon$-DP.
\end{definition}

\par
Specifically, for edge device $k$, we set the received signal $\pmb r$ as test variable to be the output of protocol $\mathcal{M}$. Based on this, $(\epsilon,\delta)$-DP guarantees that for any neighboring datasets $\mathcal{D}_k$ and $\mathcal{D}_k'$, the differential privacy loss 
\begin{equation}
\mathcal{L}_{\mathcal{D}_k,\mathcal{D}_k'}(\pmb r)=\ln\frac{\Pr(\pmb r|\mathcal{D}_k)}{\Pr(\pmb r|\mathcal{D}_k')}, \label{Privacy-loss}
\end{equation}
i.e., the log-likelihood ratio of the neighboring datasets $\mathcal{D}_k$ and $\mathcal{D}_k'$ satisfies
\begin{equation}
\Pr\left(|\mathcal{L}_{\mathcal{D}_k,\mathcal{D}_k'}(\pmb r)|\leq\epsilon\right)\geq 1-\delta. \label{Privacy-standard}
\end{equation}
From (\ref{Privacy-standard}), we can see that a lower $\epsilon$ yields a higher privacy guarantee for the FL system.

\par
To achieve the privacy level $\epsilon$, a common method is to add random perturbations to the signals. We thus uniformly assign the artificial noise $\pmb n_{k,t}\!=\![\pmb n_{k,t}^T(1),\ldots,\pmb n_{k,t}^T(I)]^T\!\in\!\mathbb{C}^d$ with $\pmb n_{k,t}(i)\sim\mathcal{CN}(\pmb 0,\sigma_{k,t}^2(i)\,\pmb I_e)$ to each communication block in the $t$-th learning round, yielding a noisy version of the local update $\pmb g_{k,t}$. Therefore, the pre-processed transmit signal of edge device $k$ is given as
\begin{equation}
\pmb s_{k,t}(i)=D_k\pmb g_{k,t}(i)+\pmb n_{k,t}(i). \label{Skt}
\end{equation}
Notice that, for simplicity, we omit some procedures such as gradient clipping \cite{Gradient-clipping} which may preclude a large value of $D_k\pmb g_{k,t}(i)$.

\subsection{Model Aggregation via Over-the-Air Computation}
We leverage AirComp for fast model aggregation by exploiting signal superposition of a wireless multiple-access channel \cite{Introduction-RIS}. In AirComp, all edge devices transmit their local updates $\{\pmb g_{k,t}(i)\}_{k=1}^{K}$ synchronously using the same time-frequency resources in the communication block, thereby achieving high communication efficiency \cite{AirComp}.

\par
Motivated by \cite{Transmit-scalar}, based on the perfect CSI of each edge device, we design the block-based uniform-forcing transmit signal $\pmb x_{k,t}(i)$ in the $i$-th communication block of learning round $t$ as follows:
\begin{equation}
\pmb x_{k,t}(i)=\alpha_{k,t}(i)\,\pmb s_{k,t}(i)=\sqrt{\eta_{t}(i)}\,\frac{\,(h_{k,t}(i))^H}{|h_{k,t}(i)|^2}\,\pmb s_{k,t}(i), \label{Transmit-signal}
\end{equation}
where $\alpha_{k,t}(i)\in\mathbb{C}$ denotes the transmit scalar which is given by $\alpha_{k,t}(i)=\sqrt{\eta_t(i)}/h_{k,t}(i)$ and $\sqrt{\eta_t(i)}>0$ is the uniform power scaling factor. Given a maximum transmit power $P_0>0$ for all edge devices, we have the following power constraints in each communication block 
\begin{equation}
\mathbb{E}\left(\left|\left|\pmb x_{k,t}(i)\right|\right|^2\right)=\mathbb{E}\left(\left|\left|\alpha_{k,t}(i)\pmb s_{k,t}(i)\right|\right|^2\right)\leq P_0,\,\forall k,i,t \label{Power-constraint}
\end{equation}
by taking the expectation over the additive Gaussian noise in (\ref{Skt}). Note that, for each communication block, $\eta_t(i)$ and $\pmb\Theta_t(i)$ are elaborately designed by the edge server and broadcasted to the edge devices and RIS, which result in additional communication overheads.

\par
Therefore, based on (\ref{Skt}) and (\ref{Transmit-signal}), the aggregated signal $\pmb r_{t}(i)$ defined in (\ref{Received-signal}) can be rewritten as
\begin{equation}
\pmb r_{t}(i)=\sqrt{\eta_t(i)}\sum_{k=1}^{K}D_k\pmb g_{k,t}(i)+\sqrt{\eta_t(i)}\sum_{k=1}^{K}\pmb n_{k,t}(i)+\pmb w_t(i). \label{Rt-vector}
\end{equation}
Then, the edge server computes the estimated global update $\hat{\pmb g}_t(i)$ by
\begin{equation}
\begin{split}
\hat{\pmb g}_{t}(i)&=\frac{1}{KD\sqrt{\eta_{t}(i)}}\text{Re}\left\{\pmb r_{t}(i)\right\} \\
&=\pmb g_{t}(i)+\frac{1}{KD}\sum_{k=1}^{K}\text{Re}\{ \pmb n_{k,t}(i)\}+\frac{\text{Re}\left\{\pmb w_t(i)\right\}}{KD\sqrt{\eta_t(i)}}, \label{Estimated-g-entry}
\end{split}
\end{equation}
where $\pmb g_{t}(i)\!=\!\sum_{k=1}^{K}\pmb g_{k,t}(i)/K$. After $I$ rounds of model aggregation, the resulting $\hat{\pmb g}_{t}\!=\![\hat{\pmb g}_{t}^T(1),\ldots,\hat{\pmb g}_{t}^T(I)]^T$ is an unbiased estimation of the true global gradient $\pmb g_{t}\!=\![\pmb g_{t}^T(1),\ldots,\pmb g_{t}^T(I)]^T$. But the existence of the artificial noise and wireless channel noise leads to inevitable inaccuracy in $\hat{\pmb g}_t$, which yields a trade-off between accuracy and privacy. In particular, the learning accuracy also depends on SNR, which is defined as the ratio of the maximum transmit power and the noise power in one communication block, i.e.,
\begin{equation}
\text{SNR}=\frac{P_0}{e\times N_0}, \label{SNR}
\end{equation}
where $e \times N_0$ represents the channel noises power within one communication block.

\par
\emph{Remark 1} (Symbol-level synchronization): Note that the AirComp-based aggregation rule relies on symbol-level synchronization among edge devices. To achieve this, a feasible technique is timing advance mechanism \cite{LTE-timing-advance1}, which is widely used in 4G Long Term Evolution (LTE) and 5G New Radio (NR) systems. Specifically, consider the case when deploying AirComp in traditional Orthogonal Frequency Division Multiplexing (OFDM) systems, we note that $1$ MHz synchronization bandwidth can limit the timing offset within $0.1$ microseconds \cite{LTE-timing-advance2}, which is shorter than the typical length of cyclic prefix (5 microseconds) in LTE systems. Thus, the symbol-level synchronization for AirComp is feasible.

\par
\emph{Remark 2} (Limitation of uniform-forcing scalar): According to (\ref{Transmit-signal}), channels inversion power control is adopted to achieve amplitude alignment at the edge server. However, due to the power constraint in (\ref{Power-constraint}), the uniform-forcing approach may be inefficient when some edge devices encounter deep fadings, which results in a small value of $\eta_t(i)$. To overcome this issue, we can adopt the truncated-channel-inversion scheme \cite{Broadband}, which allocates power to an edge device only if its channel gain exceeds a hard threshold. Besides, the deployment of RIS can also alleviate deep fadings by reconfiguring the wireless environment, which is elaborated in Section IV-A.

\subsection{Assumptions}
We list several widely used assumptions \cite{Channel-invariant,SGD-hypothesis1} for theoretical analysis of differentially private FL systems.

\begin{assumption}
        (Strong convexity) The global loss function $F(\pmb\theta)$ is strongly convex with parameter $\mu$, i.e.,
        \begin{equation}
        F(\pmb\theta)\geq F(\pmb\theta')+\nabla F(\pmb\theta')^T(\pmb\theta-\pmb\theta')+\frac{\mu}{2}||\pmb\theta-\pmb\theta'||^2,\;\forall\pmb\theta,\,\pmb\theta'. \label{PL}
        \end{equation}
\end{assumption}

\begin{assumption}
        (Smoothness) The global loss function $F(\pmb\theta)$ is Lipschitz continuous with parameter $L$, i.e.,
        \begin{equation}
        F(\pmb\theta)\leq F(\pmb\theta')+\nabla F(\pmb\theta')^T(\pmb\theta-\pmb\theta')+\frac{L}{2}||\pmb\theta-\pmb\theta'||^2,\;\forall\pmb\theta,\,\pmb\theta', \nonumber
        \end{equation}
        which implies the following inequality:
        \begin{equation}
        ||\nabla F(\pmb\theta)-\nabla F(\pmb\theta')||\leq L||\pmb\theta-\pmb\theta'||,\;\forall \pmb\theta,\,\pmb\theta'\in\mathbb{R}^d. \nonumber
        \end{equation}
\end{assumption} 

\begin{assumption}
        (Block gradient bound) In any learning round $t$, for any training data sample $(\pmb x,y)$, the sample-wise loss function is upper bounded by a given constant $\gamma_{t}$, i.e.,
        \begin{equation}
        \left|\left|\nabla f\left(\pmb x,y;\pmb\theta_{t}\right)\right|\right|\leq\gamma_{t},\;\forall t.
        \end{equation}
        Based on triangular inequality, for edge device $k$, there always exists a constant $\zeta_{k,t}(i)\leq\gamma_{t}$ satisfying $\left|\left|\pmb g_{k,t}(i)\right|\right|\leq\zeta_{k,t}(i)$.
\end{assumption}

\section{Performance Analysis and System Optimization}\label{III}
In this section, we shall analyze the convergence behavior of the RIS-enabled FL system and show the key component of the privacy-preserving mechanism. We then provide the system optimization formulation to model the dilemma between privacy and accuracy.

\subsection{Convergence Analysis}
In the $t$-th learning round, the edge server updates the global model based on the received signal $\pmb r_{t}$. However, the fading channels $\{h_{k,t}(i)\}$ as well as the artificial noise $\{\pmb n_{k,t}(i)\}$ and AWGN $\{\pmb w_t(i)\}$ may severely degrade the learning performance. To characterize the convergence behavior of the RIS-enabled FL system, we leverage the expected value of the gap between the optimal value $F^*$ and the iteration value $F(\pmb\theta_t)$ to measure the distortion over $T$ learning rounds, i.e., $\mathbb{E}[F(\pmb\theta_{T+1})-F^*]$. Hence, given $\{\eta_t(i),\sigma_{k,t}(i)\}$, we bound the average optimality gap in the following theorem.
\begin{theorem}
Under Assumptions 1 and 2, given learning rate $\lambda=1/L$, after $T$ learning iterations, the averaged optimality gap is upper bounded as
\begin{equation}
\begin{split}
\mathbb{E}&\left[F\left(\pmb\theta_{T+1}\right)\right]-F^*\leq\left(1-\frac{\mu}{L}\right)^T\left[F\left(\pmb\theta_{1}\right)-F^*\right]+\frac{1}{4L}\times \\
&\;\;\frac{e}{(KD)^2}\sum_{t=1}^{T}\left(1-\frac{\mu}{L}\right)^{T-t}\left(\sum_{k=1}^{K}\sum_{i=1}^{I}\sigma_{k,t}^2(i)+\sum_{i=1}^{I}\frac{N_0}{\eta_t(i)}\right). \label{Optimal-gap}
\end{split}
\end{equation}
\end{theorem}
\begin{proof}
Please refer to Appendix \ref{Proof-th1}.
\end{proof}
According to \pmb{Theorem 1}, we observe that the average optimality gap $\mathbb{E}[F(\pmb\theta_{T+1})\!-\!F^*]$ is upper bounded by two terms. The first term will converge to $0$ as $T$ goes to infinity, while the second term in terms of $\{\sigma_{k,t}(i),\eta_t(i)\}$ cannot vanish due to the added artificial noise and wireless channel noise in each communication block. Besides, (\ref{Optimal-gap}) is also related to $\pmb\Theta_t(i)$ as the selection of $\eta_t(i)$ depends on $\pmb\Theta_t(i)$, which will be presented in \pmb{Theorem 3}.

\subsection{Privacy-Preserving Mechanism}
We develop a privacy-preserving scheme based on DP to alleviate privacy concerns for edge devices. We assume that $\{\eta_{t}(i),\sigma_{k,t}(i),\pmb\Theta_t(i)\}$ are fixed constants, as they do not reveal any information about the local datasets. Therefore, we only focus on the upload process from edge devices to the edge server. In the $t$-th learning round, for edge device $k$, the disclosed signals regarding its local dataset $\mathcal{D}_k$ are $\{\pmb g_{k,t}(i)\}$ across $I$ blocks. We decouple the received signal $\pmb r_t(i)$ as
\begin{equation}
\pmb r_{t}(i)=\sqrt{\eta_t(i)}D_k\pmb g_{k,t}(i)+\sqrt{\eta_t(i)}\sum_{j\neq k}D_j\pmb g_{j,t}(i)+\pmb q_t(i), \label{Rt}
\end{equation}
where
\begin{equation}
\pmb q_t(i)=\sqrt{\eta_t(i)}\sum_{j=1}^{K}\pmb n_{j,t}(i)+\pmb w_t(i). \label{Effective-noise}
\end{equation}
In (\ref{Rt}), the first term is the disclosed signal concerning $\mathcal{D}_k$, the second term is independent of edge device $k$, and the third term $\pmb q_t(i)$ is Gaussian noise. We observe that the privacy-preserving disturbance in (\ref{Effective-noise}) includes the added artificial noise and the inherent channel noise. Besides, all update messages, i.e., the first and second terms in (\ref{Rt}), enjoy the same privacy protection provided by $\pmb q_t(i)$. For notational ease, let $\pmb q_t=[\pmb q_t^T(1),\ldots,\pmb q_t^T(I)]^T$ denote the effective noise to ensure privacy for the transmission of $\{\pmb g_{j,t}\}_{j=1}^{K}$ in the $t$-th learning round.

\par
For edge device $k$, its privacy level $(\epsilon,\delta)$ depends on the sensitivity of the disclosed noise-free function regarding its local dataset. We first recall the definition of $l_2$-sensitivity \cite{Differential-privacy1}.
\begin{definition}
Let $\mathcal{M}$ be an arbitrary mechanism on the transmitted signals. The $l_2$-sensitivity $\Delta$ is defined as the maximum difference of the outputs from two neighboring datasets $\mathcal{D}_k$ and $\mathcal{D}_k'$, i.e.,
\begin{equation}
\Delta=\mathop{\max}\limits_{\mathcal{D}_k,\mathcal{D}_k'}\left|\left|\mathcal{M}\left(\mathcal{D}_k\right)-\mathcal{M}\left(\mathcal{D}_k'\right)\right|\right|_2.
\end{equation}
\end{definition}
Recalling that the first term in (\ref{Rt}) is the only disclosed function concerning $\mathcal{D}_k$. Let $\pmb u_{k,t}$ denote the difference between the outputs from two neighboring datasets $\mathcal{D}_k$ and $\mathcal{D}_k'$, which is given by
\begin{equation}
\pmb u_{k,t}\!=\![h_{k,t}(1)\alpha_{k,t}(1)\Delta\pmb g_{k,t}^T(1),\!\ldots\!,\!h_{k,t}(I)\alpha_{k,t}(I)\Delta\pmb g_{k,t}^T(I)]^T, \label{Ukt}
\end{equation}
where 
\begin{equation}
\Delta\pmb g_{k,t}(i)=\sum_{(\pmb{x},y)\in\mathcal{D}_k}\nabla f_i\left(\pmb x,y;\pmb\theta_{t}\right)-\sum_{(\pmb{x},y)\in\mathcal{D}_k'}\nabla f_i\left(\pmb x,y;\pmb\theta_{t}\right) \nonumber
\end{equation}
with $\nabla f_i\left(\pmb x,y;\pmb\theta_{t}\right)$ denoting the $(i\!-\!1)e\!+\!1$-th to $ie$-th elements in $\nabla f\left(\pmb x,y;\pmb\theta_{t}\right)$. Then the $l_2$-sensitivity $\Delta_{k,t}$ of device $k$ is
\begin{equation}
\Delta_{k,t}=\mathop{\max}\limits_{\mathcal{D}_k,\mathcal{D}_k'}\left|\left|\pmb u_{k,t}\right|\right|_2. \label{Delta}
\end{equation}
By using the triangular inequality and Assumption 3, we have
\begin{equation}
\Delta_{k,t}\leq 2\gamma_t\underset{i\in[I]}{\max}\sqrt{\eta_t(i)}.\label{Delta-bound}
\end{equation}
Based on (\ref{Delta-bound}), the privacy constraint for all edge devices is given in the following theorem.
\begin{theorem}
For any fixed sequence of $\{\eta_{t}(i),\sigma_{k,t}(i)\}_{t=1}^{T}$, the RIS-enabled FL system satisfies $(\epsilon,\delta)$-DP if we have
\begin{equation}
\sum_{t=1}^{T}\frac{4\gamma_t^2}{\xi_t}\underset{i\in[I]}{\max}\left\{\eta_t(i)\right\}\leq\mathcal{R}_{dp}(\epsilon,\delta), \label{Privacy-constraints}
\end{equation}
where
\begin{equation}
\xi_t\leq\eta_t(i)\sum_{k=1}^{K}\sigma_{k,t}^2(i)+ N_0,\, \forall i,t,\label{Xi}
\end{equation}
denotes a lower bound of the power of the effective noise $\pmb q_t(i)$, and $\mathcal{R}_{dp}(\epsilon,\delta)=\big(\sqrt{\epsilon\!+\!\left[\mathcal{C}^{-1}\left(1/\delta\right)\right]^2}\!-\!\mathcal{C}^{-1}\left(1/\delta\right)\big)^2$ with $\mathcal{C}^{-1}(x)$ denoting the inverse function of $\mathcal{C}(x)=\sqrt{\pi}xe^{x^2}$.
\end{theorem}
\begin{proof}
Please refer to Appendix \ref{Proof-th2}.
\end{proof}
From \pmb{Theorem 2}, we emphasize that both the artificial noise and the receiver noise contribute to the privacy guarantee. Besides, based on (\ref{Xi}), $\xi_t$ represents a lower bound of the weakest privacy-preserving perturbation (i.e., the lowest effective noise power) across $I$ communication blocks in the $t$-th learning round. Furthermore, the privacy guarantee is related to the joint effect of $\xi_t$ and the maximum transmit power scalar $\eta_t(i)$ among $I$ communication blocks. We also notice that the relationship between the privacy constraint and the phase-shift matrix $\pmb\Theta_t(i)$ is similar to \pmb{Theorem 1}.

\subsection{System Optimization}
To design the RIS-enabled FL system, we shall propose to minimize the optimality gap in (\ref{Optimal-gap}) while satisfying the privacy and the maximum transmit power constraints given in (\ref{Privacy-constraints}) and (\ref{Power-constraint}), respectively, across $T$ learning rounds. Specifically, given the privacy level $(\epsilon,\delta)$, we jointly optimize $\{\sigma_{k,t}(i),\pmb\Theta_t(i),\eta_t(i),\xi_t\}$ in each communication block. Based on (\ref{Skt}), (\ref{Transmit-signal}) and (\ref{Power-constraint}), we obtain the power constraints,
\begin{equation}
\mathbb{E}\left(\left|\left|\pmb x_{k,t}\right|\right|^2\right)\overset{(a)}{\leq}\frac{\eta_{t}(i)}{|h_{k,t}(i)|^2}[D_k^2\zeta_{k,t}^2(i)+e\sigma_{k,t}^2(i)]\leq P_0,\,\forall i,k,t, \nonumber
\end{equation}
where (a) comes from the fact $||\pmb g_{k,t}(i)||\leq\zeta_{k,t}(i)$ in Assumption 3, along with the triangular inequality.

\par
Note that, the optimization variables are only related to the second term in (\ref{Optimal-gap}). We obtain the following optimization problem $\mathscr{P}$:
\begin{subequations} \label{P}
        \begin{align}
        & \underset{\{\eta_{t},\xi_t,\pmb\Theta_t,\sigma_{k,t}\}}{\text{minimize}}
        & &
        \sum_{t=1}^{T}\!\left(1\!-\!\frac{\mu}{L}\right)^{-t}\!\left(\sum_{k=1}^{K}\sum_{i=1}^{I}\!\sigma_{k,t}^2(i)\!+\!\sum_{i=1}^{I}\!\frac{N_0}{\eta_t(i)}\right) \label{P-obj} \\
        & \;\;\text{subject to}
        & &
        \sum_{t=1}^{T}\frac{4\gamma_t^2}{\xi_t}\underset{i\in[I]}{\max}\left\{\eta_t(i)\right\}\leq\mathcal{R}_{dp}(\epsilon,\delta), \label{P-privacy-constraint} \\
        &&&
        \xi_t\leq\eta_t(i)\sum_{k=1}^{K}\sigma_{k,t}^2(i)+N_0,\,\forall i,t, \label{P-privacy-xi} \\
        &&&
        \frac{\eta_{t}(i)}{|\pmb m_t^T(i)\pmb\Theta_t(i)\pmb h_{k,t}^r(i)\!+\!h_{k,t}^d(i)|^2}[D^2\zeta_{k,t}^2(i) \nonumber \\
        &&&
        \qquad\qquad\qquad +e\sigma_{k,t}^2(i)]\!\leq\! P_0,\,\forall i,k,t, \label{P-power-constraints} \\
        &&&
        |\pmb\Theta_t(i)(n,n)|=1,\quad\forall i,n,t, \label{P-RIS}
        \end{align}
\end{subequations}
where (\ref{P-privacy-constraint}) and (\ref{P-privacy-xi}) denote the privacy constraints for all edge devices as presented in (\ref{Privacy-constraints}) and (\ref{Xi}), respectively, (\ref{P-power-constraints}) represents the power constraints by substituting $h_{k,t}(i)=\pmb m_t^T(i)\pmb\Theta_t(i)\pmb h_{k,t}^r(i)+h_{k,t}^d(i)$ into (\ref{Power-constraint}), and (\ref{P-RIS}) indicates the unit modulus constraints of all reflecting elements at RIS with $\pmb\Theta_t(i)(n,n)$ denoting as the $(n,n)$-th entry of $\pmb\Theta_t(i)$. Unfortunately, due to the non-convexity of the privacy and power constraints, and the nonconvex unimodular constraints of the phase shifts at RIS, problem $\mathscr{P}$ is highly intractable and computationally challenging.

\section{A Two-Step Alternating Low-Rank Optimization Framework}\label{IV}
In this section, we propose a two-step alternating minimization framework for solving problem $\mathscr{P}$. Specifically, the transmit power scalar $\eta_t(i)$ and artificial noise $\sigma_{k,t}(i)$ at each edge device, as well as the phase-shift matrix $\pmb\Theta_t(i)$ at the RIS are optimized in an alternative manner until the algorithm converges.

\begin{figure*}[t]   
	\begin{minipage}{0.49\linewidth}
		\centerline{\includegraphics[scale=0.69]{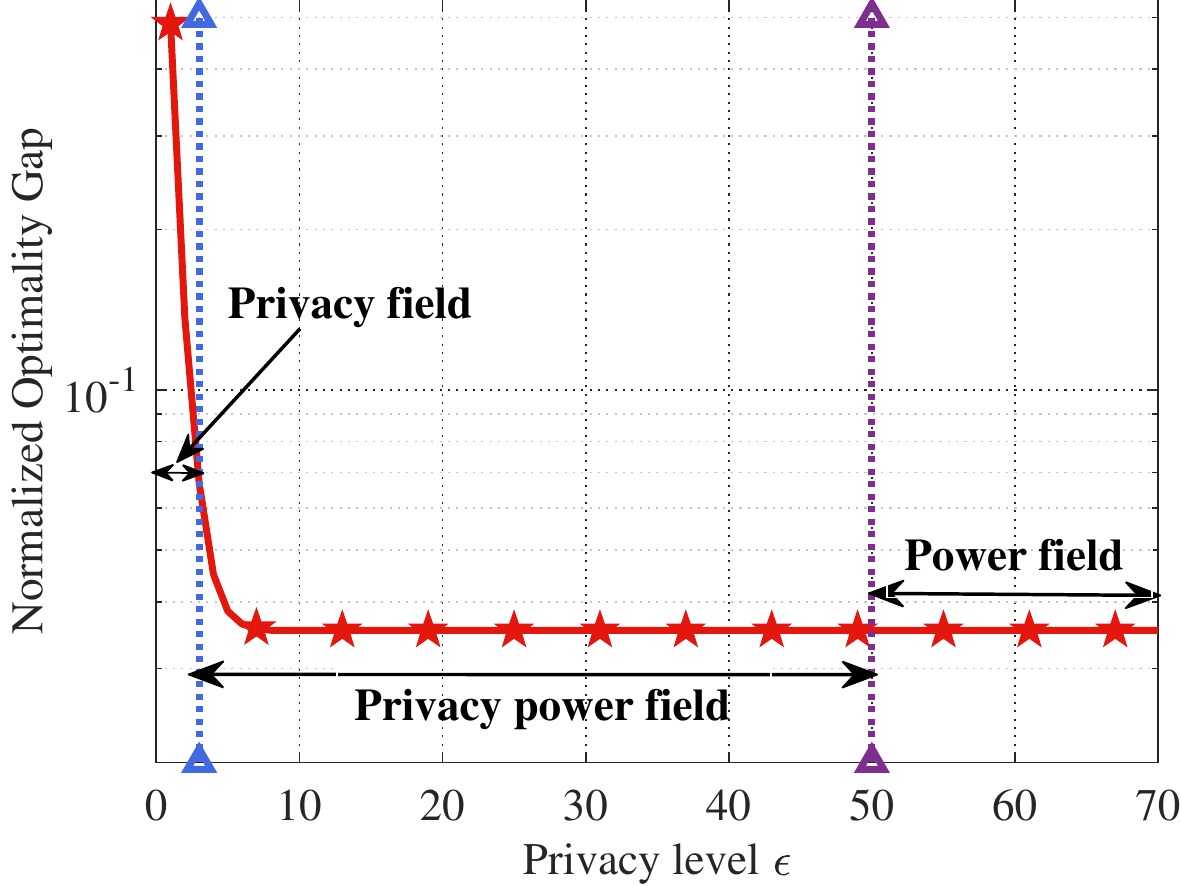}}
		\caption{The relationship between learning accuracy and privacy level. The two vertical lines distinguish different cases.}
		\label{fig3}
	\end{minipage}
	\hfill
	\begin{minipage}{0.49\linewidth}
		\centerline{\includegraphics[scale=0.39]{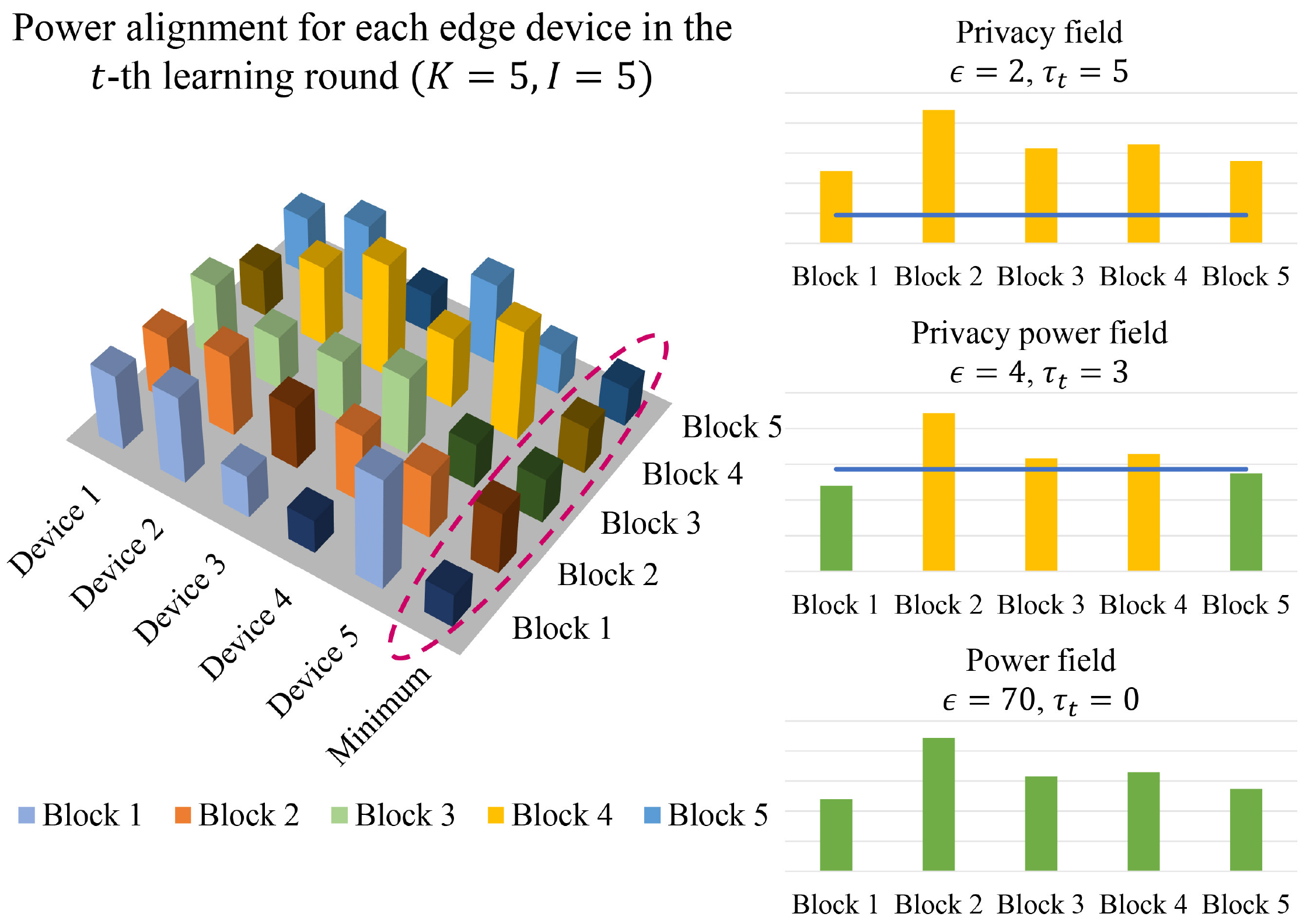}}
		\caption{Power allocation under different privacy levels in one learning round. The blue lines indicate the privacy limitations.}
		\label{fig4}
	\end{minipage}
\end{figure*}

\subsection{Co-Design of Artificial Noise and Power Scalar} \label{IV-A}
The main idea of this framework is to alternately optimize $\{\eta_t(i),\sigma_{k,t}(i),\pmb\Theta_t(i),\xi_t\}$ in $T$ learning rounds. Note that problem $\mathscr{P}$ is a non-causal setting where parameters $\{\pmb h_{k,t}^r(i),\pmb m_t(i),h_{k,t}^d(i),\zeta_{k,t}(i),\gamma_t\}$ need to be known in advance, which will be discussed in Section \ref{VI}-C.

\par
Given the phase-shift matrix $\pmb\Theta_t(i)$, problem $\mathscr{P}$ is reduced to problem $\mathscr{P}_1$ as follows:
\begin{subequations} \label{P1}
        \begin{align}
        & \underset{\{\eta_{t},\sigma_{k,t},\xi_t\}}{\text{minimize}}
        & &
        \sum_{t=1}^{T}\!\left(1\!-\!\frac{\mu}{L}\right)^{-t}\!\left(\sum_{k=1}^{K}\sum_{i=1}^{I}\!\sigma_{k,t}^2(i)\!+\!\sum_{i=1}^{I}\!\frac{N_0}{\eta_t(i)}\right) \label{P1-obj} \\
        & \text{subject to}
        & &
        \sum_{t=1}^{T}\frac{4\gamma_t^2}{\xi_t}\underset{i\in[I]}{\max}\left\{\eta_t(i)\right\}\leq\mathcal{R}_{dp}(\epsilon,\delta), \label{P1-privacy} \\
        &&&
        \xi_t\leq\eta_t(i)\sum_{k=1}^{K}\sigma_{k,t}^2(i)+N_0,\,\forall i,t, \label{P1-privacy-xi} \\
        &&&
        \frac{D^2\zeta_{k,t}^2(i)+e\sigma_{k,t}^2(i)}{|h_{k,t}(i)|^2}\eta_{t}(i)\leq P_0,\,\forall i,k,t, \label{P1-power}
        \end{align}
\end{subequations}
where $h_{k,t}(i)=\pmb m_t^T(i)\,\pmb\Theta_t(i)\,\pmb h_{k,t}^r(i)+h_{k,t}^d(i)$ denotes the composite channel response which remains constant given $\pmb\Theta_t(i)$. We note that through a change of variables, problem $\mathscr{P}_1$ can be transformed into a convex problem tackled by KKT conditions, leading to an adaptative power allocation mechanism, i.e., \pmb{Theorem 3}. For analytical ease, we denote $\tau_t$ as the number of communication blocks which are restricted by privacy in the $t$-th learning round.
\begin{theorem}
The optimal solution to problem $\mathscr{P}_1$ is given by
\begin{equation}
\xi_t=N_0,\,\forall t,\quad\sigma_{k,t}(i)=0,\,\forall i,k,t,
\end{equation}
and the selection of $\{\eta_t(i)\}$ depends on the joint effect of privacy and power, which yields the following three cases:
\begin{itemize}
\item[(a)] \emph{\pmb{Privacy field}}: If $\tau_t=I$ for all learning rounds, then the optimal $\{\eta_t(i)\}_{i=1}^I$ remains a constant $\eta_t$ with
\begin{equation}
\eta_t=\frac{N_0\sqrt{I}}{\;2\gamma_t\beta^{\frac{1}{2}}}\left(1-\frac{\mu}{L}\right)^{-\frac{t}{2}}\,\land\,\eta_t<P_0\underset{i,k}{\min}\frac{\;|h_{k,t}(i)|^2}{D^2\zeta_{k,t}^2(i)}, \label{Optimal-eta3}
\end{equation}
where $\beta$ is selected to strictly satisfy (\ref{P1-privacy}), i.e.,
\begin{equation}
\sum_{t=1}^{T}\frac{4\gamma_t^2}{N_0}\eta_t=\sum_{t=1}^{T}\frac{2\sqrt{I}\gamma_t}{\;\beta^{\frac{1}{2}}}\left(1-\frac{\mu}{L}\right)^{-\frac{t}{2}}=\mathcal{R}_{dp}(\epsilon,\delta). \label{Case-a-equ}
\end{equation} 
\item[(b)] \emph{\pmb{Privacy power field}}: It is a generalized version of (a) with $\tau_t\in[I]$, and the optimal $\eta_t(i)$ is given by
\begin{equation}
\eta_t(i)\!=\!\min\left\{\frac{N_0\sqrt{\tau_t}}{\;2\gamma_t\beta^{\frac{1}{2}}}\left(1\!-\!\frac{\mu}{L}\right)^{-\frac{t}{2}},P_0\underset{k}{\min}\frac{\;|h_{k,t}(i)|^2}{D^2\zeta_{k,t}^2(i)}\right\}, \label{Optimal-eta2}
\end{equation}
where $\beta$ and $\tau_t$ are selected to strictly satisfy (\ref{P1-privacy}), i.e.,
\begin{equation}
\begin{aligned}
\sum_{t=1}^{T}\frac{4\gamma_t^2}{N_0}\underset{i\in[I]}{\max}\left\{\eta_t(i)\right\}\!=\!\sum_{t=1}^{T}\frac{4\gamma_t^2}{N_0}\underset{i\in[I]}\min\left\{\frac{N_0\sqrt{\tau_t}}{\;2\gamma_t\beta^{\frac{1}{2}}}\Big(1-\frac{\mu}{L}\Big)^{-\frac{t}{2}}, \notag\right.
\\
\phantom{=\;\;}
\left.P_0\,\underset{i\in[I]}{\max}\left\{\underset{k\in[K]}{\min}\frac{\;|h_{k,t}(i)|^2}{D^2\zeta_{k,t}^2(i)}\right\}\right\}=\mathcal{R}_{dp}(\epsilon,\delta). \label{Case-b-equ}
\end{aligned}
\end{equation}
\item[(c)] \emph{\pmb{Power field}}: If $\tau_t=0$ holds for all learning rounds, i.e.,
\begin{equation}
\sum_{t=1}^{T}\frac{4P_0\gamma_t^2}{N_0}\underset{i\in[I]}{\max}\left\{\underset{k\in[K]}{\min}\frac{\;|h_{k,t}(i)|^2}{D^2\zeta_{k,t}^2(i)}\right\}<\mathcal{R}_{dp}(\epsilon,\delta),
\end{equation}
the unique optimal power scalar $\eta_t(i)$ is given by
\begin{equation}
\eta_t(i)=P_0\,\underset{k\in[K]}{\min}\frac{\;|h_{k,t}(i)|^2}{D^2\zeta_{k,t}^2(i)}. \label{Optimal-eta1}
\end{equation}
\end{itemize}
\end{theorem}
\begin{proof}
Please refer to Appendix \ref{Proof-th3}.
\end{proof}
\vskip -1em
From \pmb{Theorem 3}, the power of artificial noise becomes zero, which indicates that the additive channel noise in model aggregation serves as an inherent privacy-preserving mechanism to guarantee DP levels for each edge device \cite{FLPrivacy-free}. Besides, as shown in Fig. \ref{fig3}, three cases characterize the learning accuracy of the FL system. Essentially, the power allocation scheme, i.e., the selection of $\eta_t(i)$ under different $\epsilon$ causes the difference in learning accuracy. To make it precise, Fig. \ref{fig4} illustrates this process in one learning round. Specifically, two power alignment schemes are developed for comparison, where one represents the maximum transmit power concerning privacy requirement $\{\epsilon,\tau_t\}$, and the other indicates the case without privacy, i.e., the Minimum column in Fig. \ref{fig4}. We note that $\tau_t$, which represents the impact of $\epsilon$ across $I$ communication blocks, is the key component of \pmb{Theorem 3}. Concretely, $\tau_t$ indicates the number of communication blocks restricted by privacy in one learning round, i.e., the number of yellow bars in Fig. \ref{fig4}. Consequently, different $\epsilon$ results in a different selection of $\tau_t$, yielding the following three cases:
\begin{itemize}
\item Case (a): Due to the extremely strict privacy level, all communication blocks in this case are restricted by privacy, i.e., $\tau_t=I$.
\item Case (b): In this case, several communication blocks are restricted by privacy while others are limited by transmit power, i.e., $\tau_t\in[I]$. Besides, cases (a) and (b) must satisfy the equality in (\ref{P-privacy-constraint}), which is achieved by adaptively selecting $\{\beta,\tau_t\}$ through enumeration.
\item Case (c): In this case, due to the quite loose privacy level, all communication blocks are limited by the transmit power constraints, i.e., $\tau_t=0$.
\end{itemize}

\par
We reveal the benefits achieved by RIS to the above three cases. Due to the reconfigurable capability of RIS, we note that the RIS-enabled FL system is able to establish better channel conditions compared to the FL systems without RIS. Besides, according to (\ref{Optimal-eta2}) and (\ref{Optimal-eta1}), the RIS-enabled FL system can enjoy higher transmit power and enhance the power of received signals at the edge server, thereby resulting in high learning accuracy. However, when the privacy level is strict, i.e., Case (a), the learning accuracy is restricted by privacy and cannot be improved by RIS. We will further explore the significant impacts brought by RIS on the learning accuracy via simulations in Section \ref{V}.

\subsection{Design of Phase-Shift Matrix}
On the other hand, for given the artificial noise and power scalar $\{\sigma_{k,t}(i), \eta_t(i)\}$, problem $\mathscr{P}$ becomes the following nonconvex feasibility detection problem $\mathscr{P}_2$:
\begin{subequations} \label{P2}
        \begin{align}
        & \quad \text{find}
        & &
        \pmb\Theta_t(i) \label{P2-ori-obj} \\
        & \text{subject to}
        & &
        \frac{\eta_{t}(i)\,D^2\zeta_{k,t}^2(i)}{|\pmb m_t^T(i)\pmb\Theta_t(i)\pmb h_{k,t}^r(i)\!+\!h_{k,t}^d(i)|^2}\!\leq\! P_0,\,\forall i,k,t, \\  
        &&&
        |\pmb\Theta_t(i)(n,n)|=1,\,\forall i,n,t.
        \end{align}
\end{subequations}
For analytical ease, by denoting $v_{n,t}(i)=e^{j\phi_{n,t}(i)}$ and $\pmb c_{k,t}^H(i)=\pmb m_t^T(i)\,\text{diag}(\pmb h_{k,t}^r(i))$, we have $\pmb m_t^T(i)\,\pmb\Theta_t(i)\,\pmb h_{k,t}^r(i)\!=\!\pmb c_{k,t}^H(i)\,\pmb v_t(i)$, where $\pmb v_t(i)\!=\![v_{1,t}(i),\ldots,v_{N,t}(i)]^T$. Therefore, problem (\ref{P2}) is further transformed into problem $\mathscr{P}_{2.1}$:
\begin{subequations} \label{P2.1}
        \begin{align}
        & \quad \text{find}
        & &
        \pmb v_t(i) \label{P2-obj} \\
        & \text{subject to}
        & &
        \frac{\eta_{t}(i)D^2\zeta_{k,t}^2(i)}{|\pmb c_{k,t}^H(i)\pmb v_t(i)+h_{k,t}^d(i)|^2}\leq P_0,\,\forall k,t, \label{P2-power} \\
        &&&
        |v_{n,t}(i)|=1,\,\forall n\in[N].
        \end{align}
\end{subequations}
Nevertheless, problem $\mathscr{P}_{2.1}$ is still nonconvex and inhomogeneous. To develop an efficient algorithm, by introducing an auxiliary variable $\iota_t(i)$, problem (\ref{P2.1}) can be equivalently reformulated as a homogeneous nonconvex quadratically constrained quadratic programming (QCQP) problem \cite{RIS-principle,Approximating}, which is given by the following problem $\mathscr{P}_{2.2}$:
\begin{subequations} \label{P2.2}
        \begin{align}
        & \quad \text{find}
        & &
        \hat{\pmb v}_t(i) \\
        & \text{subject to}
        & &
        \hat{\pmb v}_t^H(i)\pmb R_{k,t}(i)\hat{\pmb v}_t(i)+|h_{k,t}^d(i)|^2 \nonumber \\
        &&&
        \qquad\qquad\quad\geq\eta_{t}(i)D^2\zeta_{k,t}^2(i)\big/P_0,\,\forall i,k,t, \\
        &&&
        |\hat{v}_{n,t}(i)|=1,\,\forall n\in[N+1],
        \end{align}
\end{subequations}
where
$$
\pmb R_{k,t}(i)=
\begin{bmatrix}
\pmb c_{k,t}(i)\left(\pmb c_{k,t}(i)\right)^H & \pmb c_{k,t}(i)h_{k,t}^d(i) \\
\left(\pmb c_{k,t}(i)h_{k,t}^d(i)\right)^H & 0
\end{bmatrix}
,\;
\hat{\pmb v}_t(i)=
\begin{bmatrix}
\pmb v_t(i) \\ \iota_t(i)
\end{bmatrix}.
$$
Let $\hat{\pmb v}_t^*(i)=[\widetilde{\pmb v}_t^*(i)^T,\iota_t^*(i)]^T$ denote a feasible $\hat{\pmb v}_t(i)$ to problem (\ref{P2.2}). Then a feasible solution $\pmb v_t^*(i)$ to problem (\ref{P2.1}) can be immediately recovered by setting $\pmb v_t^*(i)=\widetilde{\pmb v}_t^*(i)/\iota_t^*(i)$. A feasible solution $\pmb\Theta_t^*(i)$ to problem (\ref{P2}) can thus be expressed as $\pmb\Theta_t^*(i)=\text{diag}\,(\pmb v_t^*(i))$.

\par
To solve problem (\ref{P2.2}), a natural method is to formulate it as a semi-definite programming (SDP) problem by matrix lifting \cite{AirComp,Matrix-lifting}. Since $\hat{\pmb v}_t^H(i)\,\pmb R_{k,t}(i)\,\hat{\pmb v}_t(i)=\trace(\pmb R_{k,t}(i)\,\hat{\pmb v}_t(i)\,\hat{\pmb v}_t^H(i))$, we denote $\pmb V_t(i)$ as the lifting matrix of $\hat{\pmb v}_t(i)$, where $\pmb V_t(i)=\hat{\pmb v}_t(i)\,\hat{\pmb v}_t^H(i)$ and $\text{rank}\,(\pmb V_t(i))=1$. Therefore, problem $\mathscr{P}_{2.1}$ can be further transformed into the following problem $\mathscr{P}_{2.3}$:
\begin{subequations}
        \begin{align}
        & \quad \text{find}
        & &
        \pmb V_t(i) \\
        & \text{subject to}
        & &
        \trace\left(\pmb R_{k,t}(i)\pmb V_t(i)\right)+|h_{k,t}^d(i)|^2 \nonumber \\
        &&&
        \qquad\qquad\geq\eta_{t}(i)D^2\zeta_{k,t}^2(i)\big/P_0,\,\forall i,k,t, \\
        &&&
        \pmb V_t(i)(n,n)=1,\,\forall n\in[N+1], \\ 
        &&&
        \text{rank}\left(\pmb V_t(i)\right)=1,\,\forall i,t, \label{Rank-one} \\
        &&&
        \pmb V_t(i)\succeq\pmb 0,\,\forall i,t,
        \end{align}
\end{subequations}
where $\pmb V_t(i)(n,n)$ denotes the $(n,n)$-entry of matrix $\pmb V_t(i)$.

\par
However, the resulting problem $\mathscr{P}_{2.3}$ is still nonconvex due to the rank-one constraints in (\ref{Rank-one}). Fortunately, a feasible solution to this problem can be generated by simply dropping the rank-one constraints via the SDR technique \cite{Matrix-lifting}. The resulting SDP problem can be solved efficiently by existing convex optimization solvers such as CVX \cite{CVX}. Although the Gaussian randomization technique in SDR may generate a suboptimal solution, we can still observe that a significant learning performance gain of the RIS-enabled FL system can be achieved through numerical experiments. 

\begin{algorithm}[h]
	\caption{Two-step alternating minimization for solving problem $\mathscr{P}$ (\ref{P}) in the RIS-enabled FL system.}
	\KwIn{The phase-shift matrix $\pmb\Theta_0$, the privacy level $(\epsilon,\delta)$, and the maximum iteration number $J$.}
	\BlankLine
	Initialize $j\gets 1$ to denote the number of iterations. \\
	\Repeat{(\ref{P-obj}) is below a given threshold or $j>J$}{
		Given $\pmb\Theta_{j-1}$, obtain solution $(\eta_j,\sigma_{k,j})$ by solving problem $\mathscr{P}_1$ (\ref{P1}). \\
		Given $(\eta_j,\sigma_{k,j})$, obtain solution $\pmb\Theta_j$ by solving problem $\mathscr{P}_2$ (\ref{P2}). \\
		Update $j\gets j+1$.}
	\BlankLine
	\KwOut{$\eta_t(i)\gets\eta_j$, $\pmb\Theta_t(i)\gets\pmb\Theta_j$, $\sigma_{k,t}(i)\gets\sigma_{k,j}$.}
\end{algorithm}

\par
Without causing confusion, we omit parameter $i$ for notational ease, e.g., $\pmb\Theta_t(i)$ as $\pmb\Theta_t$. We present the two-step iterative framework in \pmb{Algorithm 1} for solving problem $\mathscr{P}$. The computational cost of the proposed algorithm consists of solving a sequence of search programs for $\mathscr{P}_1$ and SDR programs \cite{Matrix-lifting} for $\mathscr{P}_2$. From the complexity analysis of typical interior-point method, the worst-case complexity of \pmb{Algorithm 1} is obtained as $\mathcal{O}\left(JIT\max\{K,N+1\}^4\sqrt{N+1}\log_2(1/o)\right)$, where $o>0$ denotes the solution accuracy of SDR.

\subsection{Discussion on the Causal Approximation}
Based on the above observations, \pmb{Algorithm 1} solves the nonconvex problem $\mathscr{P}$ with the knowledge of all channel responses $\{h_{k,t}^d(i),\pmb m_t(i),\pmb h_{k,t}^r(i)\}$ and local gradient information $\{\zeta_{k,t}(i),\gamma_t\}$, yielding an impractical non-causal system. We thus present some feasible prediction methods to approximate this non-causal system.

\par
At the beginning of the training process, we need to address the following two questions for our proposed scheme: 
\begin{itemize}
\item How to approximate the future gradient information?
\item How to predict the future channel responses?
\end{itemize}
The first question aims at providing upper bounds for the $l_2$-sensitivity defined in (\ref{Delta}) and the local gradients. To this end, there exist multiple adaptive methods for practical use, such as the $L$-Lipschitz based technique \cite{L-bound1,L-bound2}, and the adaptive clipping method \cite{Clipping-gradient} with a well-designed critical threshold based on the current gradients.

\par
For the second question, this is essentially related to wireless channel prediction. To estimate channels accurately, many promising methods have been proposed. For instance, the authors in \cite{FLPrivacy-free} specialized an accumulated mechanism where the future channel responses are assumed to be consistent with the current ones. But this assumption fails to exploit the correlation among channel blocks. To explore the inherent relationship between each communication block, several prediction techniques have been developed in \cite{Channel-prediction1,Channel-prediction2}, including the Gauss-Markov process channel modeling method \cite{Channel-prediction-Kalman2} and Wiener or Kalman filtering based approach \cite{Channel-prediction-Kalman1,Channel-prediction-Kalman2}. Besides, the deep learning methods turn out to be powerful to improve channel prediction accuracy \cite{Channel-prediction-ML1,Channel-prediction-ML2,Deep-learning-CSI}. We thus leave the design of effective casual models as our future work.

\section{Simulation Results} \label{V}
In this section, we present the simulation results to demonstrate the advantages of the RIS-enabled FL systems. Besides, the effectiveness of our proposed two-step alternating minimization framework will also be illustrated. Simulations are conducted using Matlab R2021b and the code is available at https://github.com/MengCongWo/FL\_Privacy\_blockcrossing.

\subsection{Simulation Settings}
We propose to gain insights into the impact of deploying RIS on learning accuracy based on ridge regression, whose sample-wise loss function is given as
\begin{equation}
f(\pmb x,y;\pmb\theta)=\frac{1}{2}\left|\left|\pmb\theta^T\pmb x-y\right|\right|^2+\upsilon||\pmb\theta||^2, \label{Ridge-regression}
\end{equation}
where $\upsilon=5\times 10^{-5}$ denotes penalty coefficient. We randomly generate a dataset of scale $|\mathcal{D}|=10^4$ and set the model dimension $d$ to be $10$. Specifically, the training samples $\pmb x$ are drawn i.i.d. according to $\mathcal{N}(\pmb 0,\pmb I_d)$ while the corresponding true label $y$ is given by $y=\pmb x(2)+3\pmb x(5)+0.2z_0$, where the data noise $z_0$ is drawn i.i.d. from $\mathcal{N}(0,1)$. We uniformly divide the total dataset $\mathcal{D}$ into $K$ local datasets, and let the ratio
\begin{equation}
r=\underset{k\in[K]}{\max}\frac{|\mathcal{D}_k|}{|\mathcal{D}|}\in[0.1,1)
\end{equation}
denote the system heterogeneity representing various storage capacities of the edge devices and larger $r$ yields higher heterogeneity. The global loss function $F$ is $\mu$-strongly convex, $L$-Lipschitz smooth and differentiable, where $\mu$ and $L$ are specified by the smallest and largest eigenvalues of the data Gramian matrix $\pmb X^T\pmb X/D+2\upsilon\pmb I_d$ with the data matrix $\pmb X\!=\![\pmb x_1,\ldots,\pmb x_D]$. Besides, the optimal $\pmb\theta^*$ of (\ref{Ridge-regression}) is $\pmb\theta^*\!=\!(\pmb X^T\pmb X+2D\upsilon\pmb I_d)^{-1}\pmb X^T\pmb y$ with true label vector $\pmb y=[y_1,\ldots,y_D]^T$. As for the definitions of $\gamma_t$ and $\zeta_{k,t}(i)$, we use the simple upper bounds as in \cite[Section V-A]{FLPrivacy-free}. Furthermore, we use the normalized optimality gap defined as $[F(\pmb\theta_{T+1})-F(\pmb\theta^*)]/F(\pmb\theta^*)$ to measure the learning accuracy.

\par
The wireless channels are assumed to suffer from Rice fading \cite{RIS-principle}, and the channel coefficients are given by
\begin{equation}
\pmb\varrho=\sqrt{\frac{\kappa}{1+\kappa}}\pmb\varrho_{LoS}+\sqrt{\frac{1}{1+\kappa}}\pmb\varrho_{NLoS}, \label{Rice-fading}
\end{equation}
where $\kappa$ represents the Rician factor, $\pmb\varrho_{LoS}$ denotes the deterministic line-of-sight (LoS) component, and $\pmb\varrho_{NLoS}$ denotes the non-line-of-sight (NLoS) component. For simplicity, we set $\pmb\varrho_{LoS}=1$ and generate $\pmb\varrho_{NLoS}$ by autoregressive (AR) scheme, which is defined as
\begin{equation}
\pmb\varrho_{NLoS}(i)=\rho\pmb\varrho_{NLoS}(i-1)+\sqrt{1-\rho^2}\pmb\vartheta(i),
\end{equation}
where $\rho$ denotes the correlation coefficient  and $\pmb\vartheta(i)$ is drawn based on an innovation process satisfying $\pmb\vartheta(i)\sim\mathcal{N}(\pmb 0,\pmb I)$. Hence, the channel coefficients in the $i$-th communication block of learning round $t$ are given by $\pmb m_t(i)=\pmb\varrho_{IB}$, $\pmb h_{k,t}^r(i)=\pmb\varrho_{DI}$, and $h_{k,t}^d(i)=\varrho_{DB}$, whose Rician factors are given by $\kappa_{IB}$, $\kappa_{DI}$, and $\kappa_{DB}$, respectively. The parameter $\rho$ is set to $1$ for simplicity since $\rho$ has no discernible effect on the performance in the case where perfect CSI is available for all edge devices.

\par
Under the above simulation settings, we mainly compare the learning accuracy under the following three RIS schemes:
\begin{itemize}
\item \pmb{RIS-enabled FL system with DP}: In this scheme, Algorithm 1 is applied to solve problem $\mathscr{P}$.
\item \pmb{RIS-enabled FL system without DP}: In this scheme, we remove the privacy constraints for the baseline without DP, i.e., only the second term in (\ref{Optimal-eta2}) exists.
\item \pmb{FL system without RIS}: In this scheme, we set $\pmb\Theta_t(i)=\pmb 0$ for the baseline without RIS but with DP \cite{FLPrivacy-free}.
\end{itemize}

\begin{table}[h]\footnotesize	
	\renewcommand\arraystretch{1.2}	
	\vspace*{-0.6em}
	\centering 
	\begin{tabular}{|>{\color{black}}c |>{\color{black}}c||>{\color{black}}c|>{\color{black}}c||>{\color{black}}c|>{\color{black}}c|}
		\hline
		\textbf{Parameter} & \textbf{Value} & \textbf{Parameter} & \textbf{Value} &
		\textbf{Parameter} & \textbf{Value} \\
		\hline
		$K$ & $10$ & $N$ & $30$ & $I$ & $5$ \\
		\hline
		$D_k$ & $1000$& $r$ & $0.1$ & $T$ & $30$ \\
		\hline
		$\kappa_{IB}$ & $5$ & $\kappa_{DI}$ & $0$ & $\kappa_{DB}$ & $5$ \\
		\hline
		$\epsilon$ & $20$ & $\delta$ & $0.01$ & $\text{SNR}$ & $35$ \\
		\hline
	\end{tabular}
	\caption{System parameters.}
	\label{table2}
\end{table}
\vskip 1mm

\begin{figure*}[t] 
	\begin{minipage}{0.49\linewidth}
		\centerline{\includegraphics[scale=0.66]{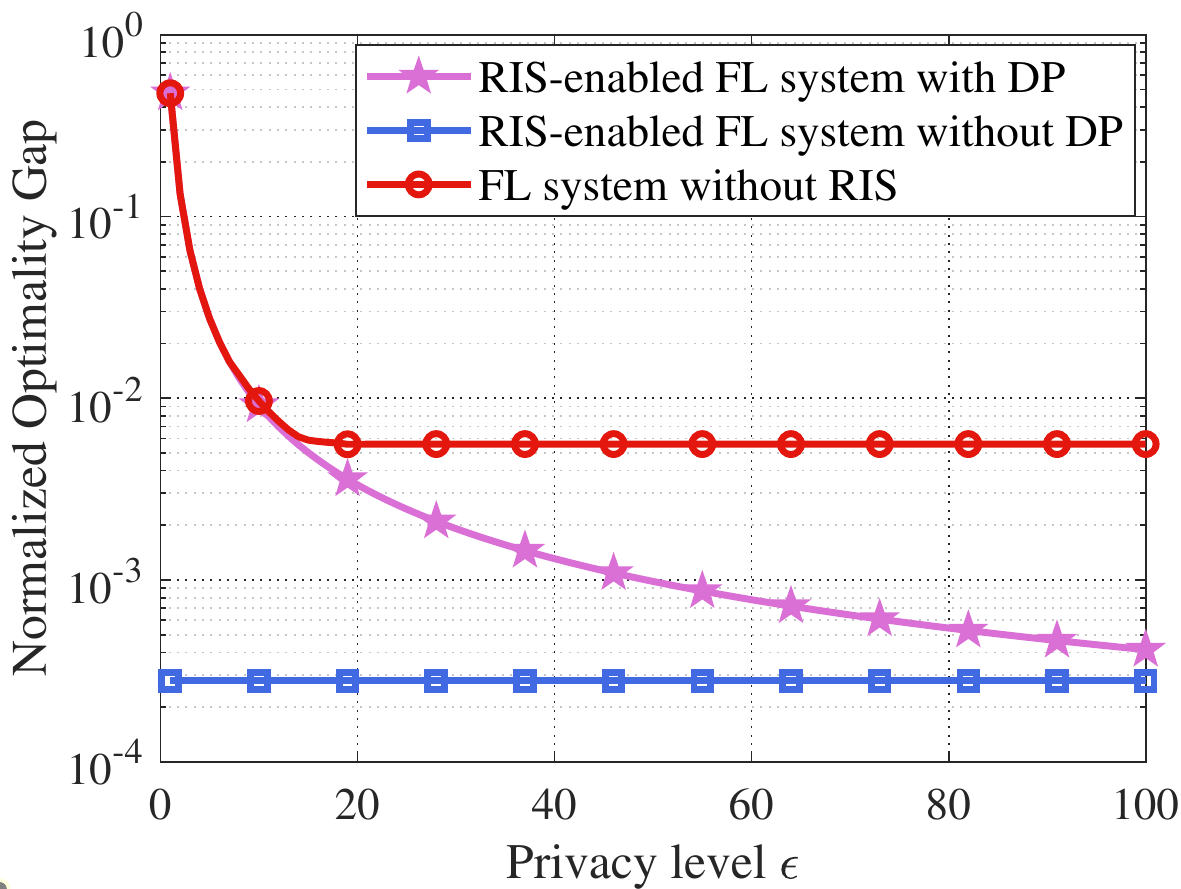}}
		\caption{Learning accuracy of three RIS schemes under different privacy levels.}
		\label{Sim-privacy-level}
	\end{minipage}
	\hfill
	\begin{minipage}{0.49\linewidth}
		\centerline{\includegraphics[scale=0.64]{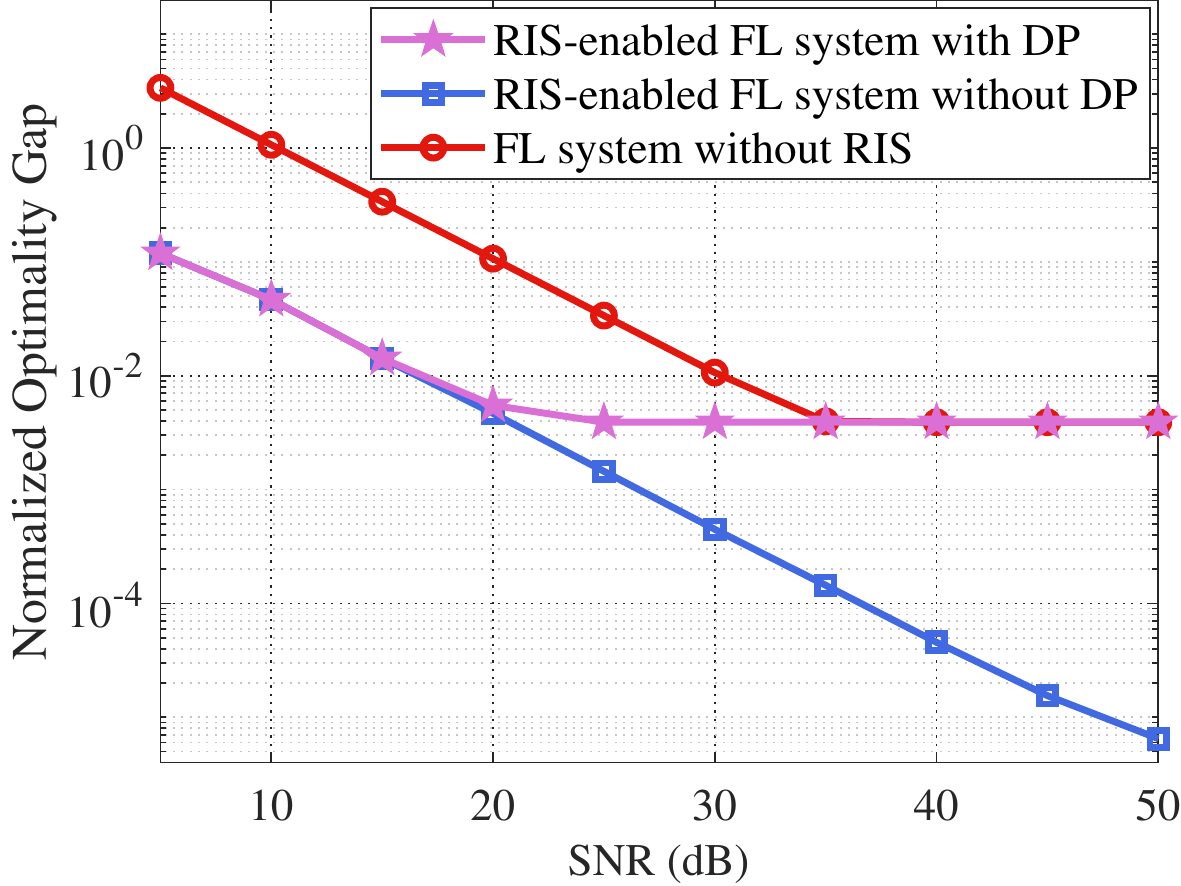}}
		\caption{Learning accuracy of three RIS schemes under various noisy environments.}
		\label{Sim-SNR-level}
	\end{minipage} 
	\vskip 20pt
	\begin{minipage}{0.49\linewidth}
		\centerline{\includegraphics[scale=0.65]{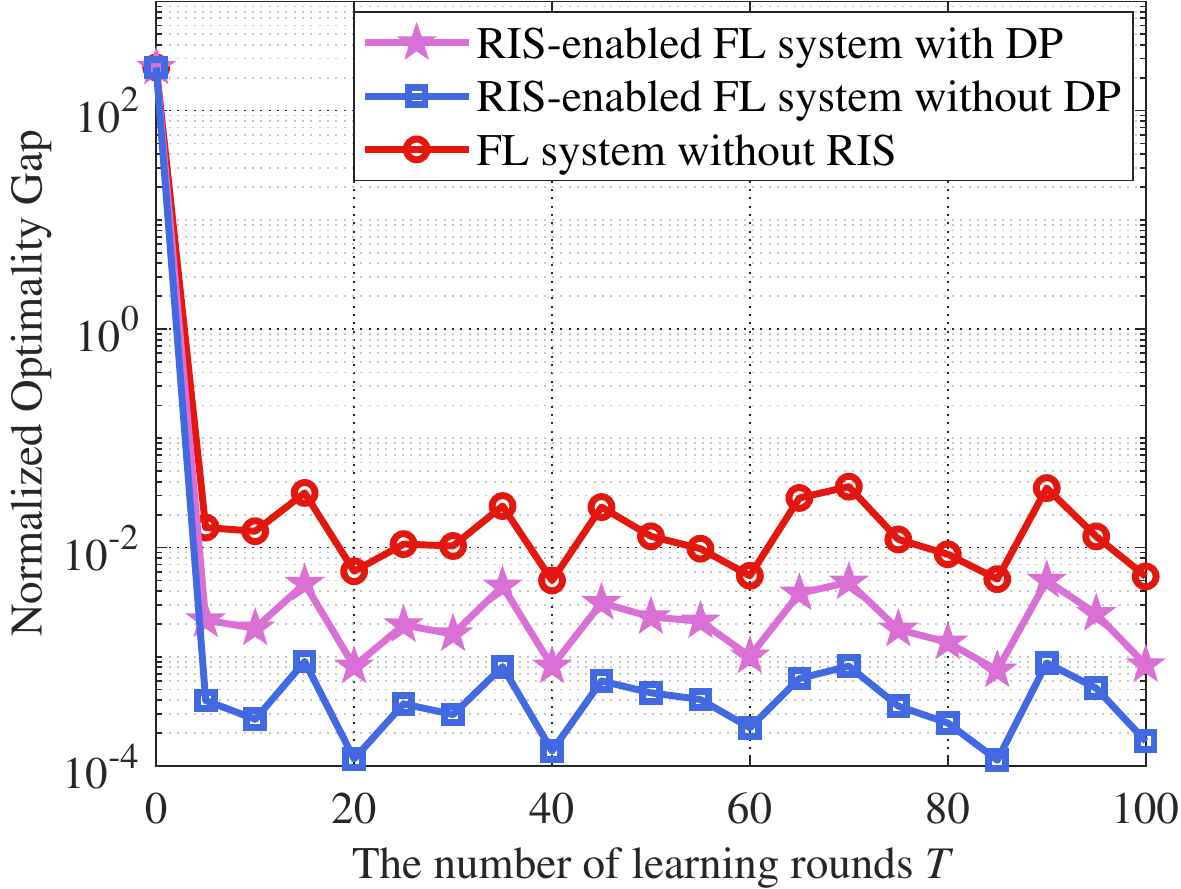}}
		\caption{Learning accuracy of three RIS schemes under different numbers of learning rounds.}
		\label{Sim-learning-round-vn}
	\end{minipage}
	\hfill
	\begin{minipage}{0.49\linewidth}
		\centerline{\includegraphics[scale=0.65]{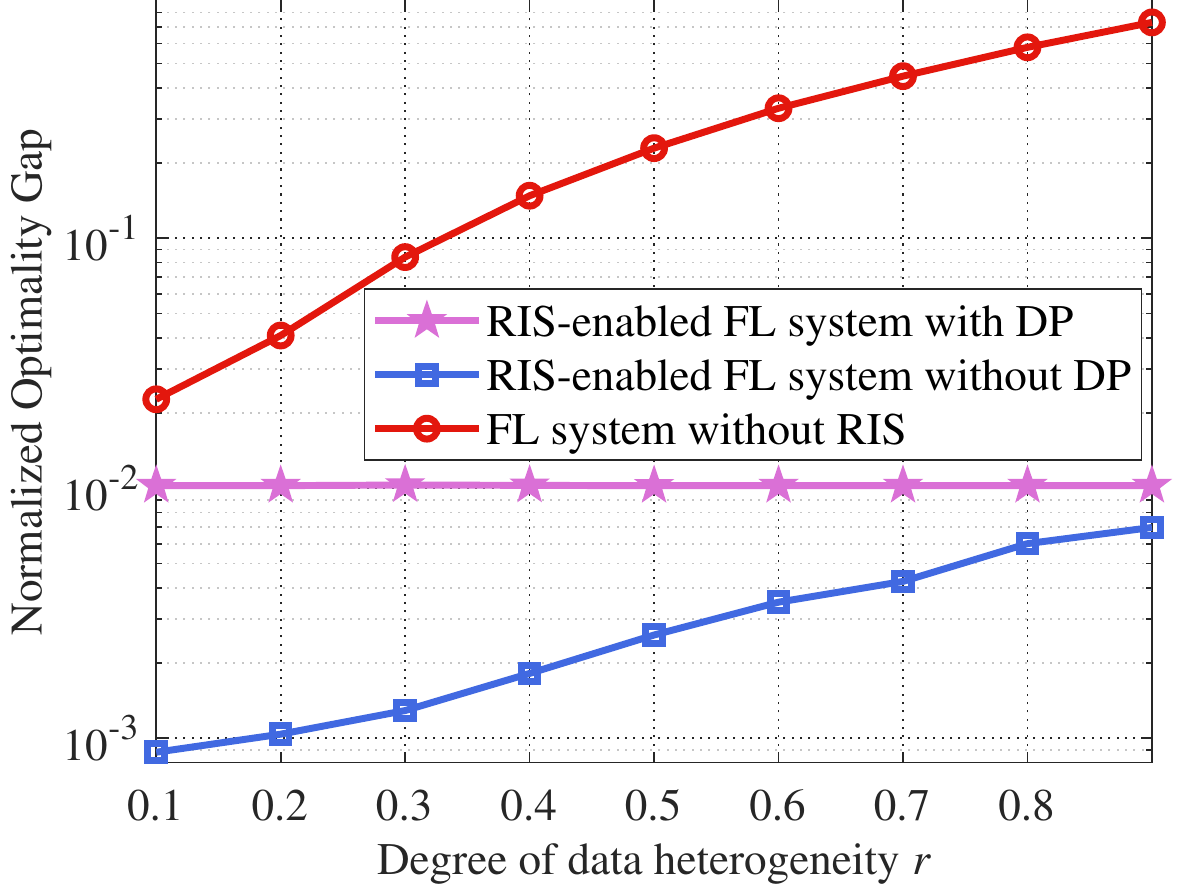}}
		\caption{Learning accuracy of three RIS schemes under various degrees of heterogeneity.}
		\label{Sim-heterogeneity}
	\end{minipage}
\end{figure*}

\par
The system parameters used in the simulations are summarized in Table \ref{table2} for reference. We also simulate a practical IoT-based affective computing FL scenario \cite{Wearable} with high dimensional and nonconvex settings, which will be presented in Section \ref{V}-C.

\subsection{Learning Accuracy under Different Conditions}
To investigate the improvements of learning accuracy achieved by RIS, we compare the learning accuracy of the three RIS schemes based on the ridge regression model from five aspects: privacy levels $\epsilon$, SNR levels, the number of learning rounds $T$, the impact of system heterogeneity $r$, and the number of RIS elements $N$. We also simulate a high dimensional setting based on the CIFAR-10 dataset and a nonconvex setting based on the MNIST dataset.

\par
We first illustrate the relationship between learning accuracy and privacy level $\epsilon$ in Fig. \ref{Sim-privacy-level}. As an extension of Fig. \ref{fig3}, we demonstrate that with the relaxation of privacy level, i.e., for larger values of $\epsilon$, the RIS-enabled FL system achieves a higher learning performance gain compared to the case without RIS which is restricted by unfavorable wireless channel propagations. Instead, RIS has no impact on the learning performance when the privacy level is strict, i.e., the privacy field presented in \pmb{Theorem 3}. Besides, comparing the learning accuracy of the RIS-enabled FL system with DP and the case without DP, we note that the privacy guarantee is achieved at the expense of losing learning accuracy, i.e., higher privacy (smaller value of $\epsilon$) indicates lower learning accuracy.

\par
In Fig. \ref{Sim-SNR-level}, we focus on the learning accuracy of the RIS-enabled FL systems in the noisy wireless environment, i.e., various levels of SNR. It shows that, with the increase of SNR, the DP-restricted schemes achieve the same learning accuracy, i.e., around $4\times10^{-3}$, under the same privacy and power constraints. Besides, we note that the RIS-enabled FL systems achieve a notable improvement in learning accuracy compared to the FL systems without RIS, especially when SNR is low. This indicates that the RIS-enabled FL systems are more adaptable to noisy wireless channels. Furthermore, Fig. \ref{Sim-SNR-level} shows that the RIS-enabled FL system with DP achieves close accuracy with the case without DP when $\text{SNR}<20$. This indicates that privacy guarantee can be ensured freely \cite{FLPrivacy-free}.

\par
The impact brought by the number of learning rounds $T$ or total communication blocks on the learning accuracy is illustrated in Fig. \ref{Sim-learning-round-vn}, where the oscillation comes from the newly introduced wireless noise in each communication block. We observe that without considering the effect of noise, the number of learning rounds $T$ has no evident effect on the accuracy after running enough number of training rounds for gradient descent due to the adaptive selection of $\{\beta,\tau_t\}$ performed in \pmb{Theorem 3}. To make it precise, $\beta$ can flexibly adjust the selection of $\eta_t(i)$ to ensure that it satisfies all constraints with the increase of $T$. Moreover, the RIS-enabled FL systems achieve higher accuracy and better robustness to the newly added wireless channel noise.

\par
Fig. \ref{Sim-heterogeneity} shows the learning accuracy of three RIS schemes among various degrees of system heterogeneity $r$. For simplicity, one edge device is assigned more data points while the rest are evenly distributed to the remaining edge devices \cite{FLPrivacy-free}. Simulation results show the negative influence brought by the high-heterogenous system on accuracy. Besides, compared to the FL systems without RIS, the RIS-enabled FL systems are more robust due to the interaction between channel condition and system heterogeneity in (\ref{Optimal-eta2}). Furthermore, Fig. \ref{Sim-heterogeneity} serves as a promising result for the design of anti-heterogeneous FL systems by deploying RIS to weaken heterogeneity.

\begin{figure}
	\begin{minipage}{\linewidth}
		\centerline{\includegraphics[scale=0.65]{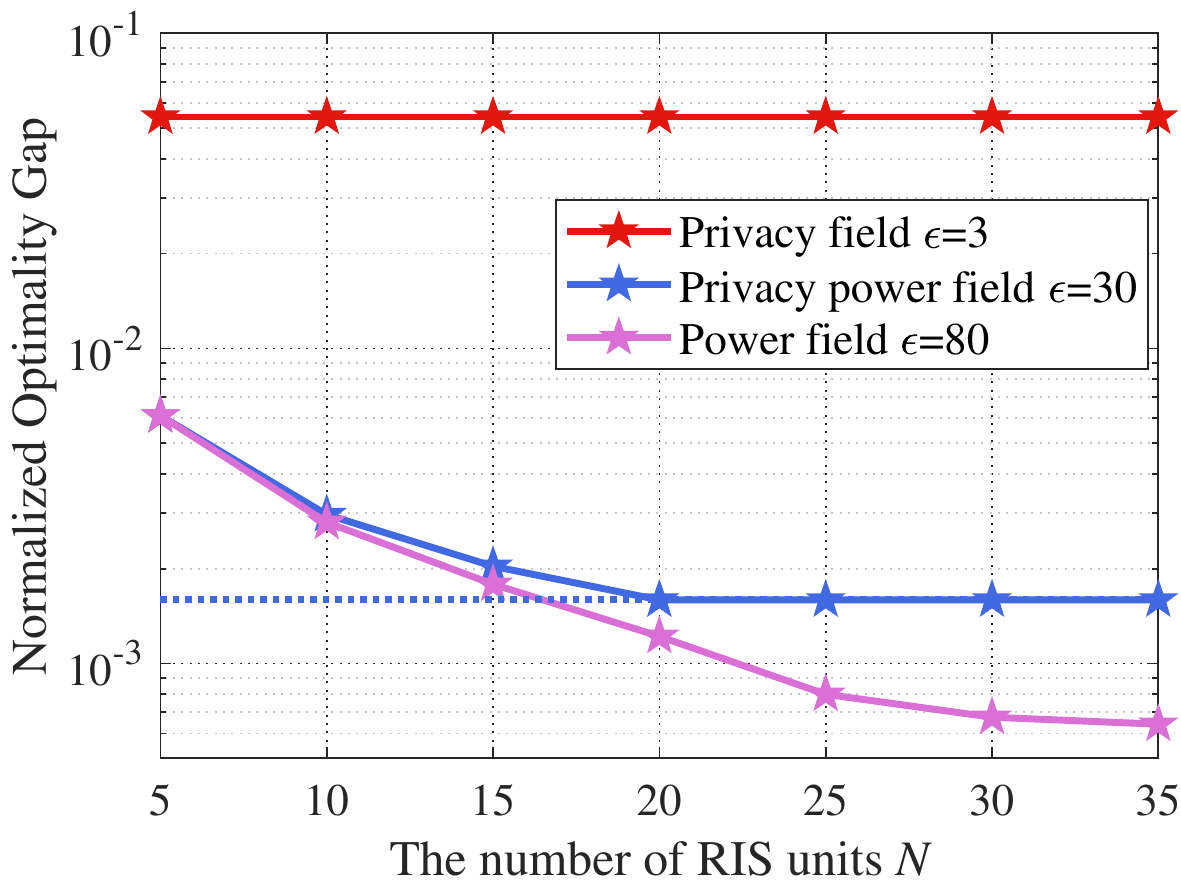}}
		\caption{Learning accuracy of the RIS-enabled FL system with DP under different privacy level $\epsilon$.}
		\label{Sim-RIS-size}
	\end{minipage}
\end{figure}

\begin{table*}[t]
	\centering
	\begin{spacing}{1.3}
		\footnotesize
		\arrayrulecolor{black}
		\begin{tabular}{|c|c|c|c|c|}
			\hline
		    & \makecell[c]{High privacy, Low SNR \\ $(\epsilon=10,\,\text{SNR}=3\,\text{dB})$} & \makecell[c]{High privacy, High SNR \\ $(\epsilon=10,\,\text{SNR}=10\,\text{dB})$} & \makecell[c]{Low privacy, Low SNR \\ $(\epsilon=80,\,\text{SNR}=3\,\text{dB})$} & \makecell[c]{Low privacy, High SNR \\ $(\epsilon=80,\,\text{SNR}=10\,\text{dB})$} \\
			\hline
			RIS-enabled FL system with DP & 0.7308 & 0.7308 & 0.7620 & 0.7716  \\
			\hline
			RIS-enabled FL system without DP & 0.7692 & 0.7716 & 0.7692 & 0.7716  \\
			\hline
			FL system without RIS & 0.6827 & 0.7260 & 0.6851 & 0.7428  \\
			\hline
			\end{tabular}
	\end{spacing}
	\caption{Classification accuracy of various FL systems for stress detection dataset under different privacy and SNR levels.}
	\label{table3}
\end{table*}

\par
Fig. \ref{Sim-RIS-size} illustrates the impacts of the number of reflecting elements $N$ at RIS on the learning accuracy under three fields presented in \pmb{Theorem 3}. From Fig. \ref{Sim-RIS-size}, we verify that the learning accuracy cannot be significantly improved when the privacy level is extremely strict (privacy field), while the benefits of deploying RIS emerge when the privacy level is relaxing. Besides, we note that with the increase of $N$, the reconfigurable capability of RIS enhances, yielding high learning accuracy. However, due to the existence of privacy constraints, this improvement of learning accuracy may slow down when $N$ becomes large.

\par
Fig. \ref{MNIST} and Fig. \ref{CIFAR10} show the learning accuracy of three RIS schemes on the MNIST and CIFAR-10 datasets, respectively. Specifically, we train a logistic regression model on the CIFAR-10 dataset for high dimensional settings and a neural network on MNIST for nonconvex settings. The neural network model parameter $\pmb\theta$ with dimension $d\!=\!79510$ is comprised of the first-layer parameter $\pmb W_1\!\in\!\mathbb{R}^{100\times 785}$ and the second-layer parameter $\pmb W_2\!\in\!\mathbb{R}^{10\times 101}$. For high-privacy and noisy environment, we reset the privacy level $\epsilon$ to be $10$, the SNR level to be $5\,\text{dB}$, the number of communication blocks in one learning round $I$ to be $10$, and the Rice factor of $\pmb\varrho_{DB}$ to be $0$ (Rayleigh channel). Fig. \ref{MNIST} and Fig. \ref{CIFAR10} indicate that the RIS-enabled FL system with DP enjoys higher accuracy and better robustness to the wireless channel noise compared to the FL system without RIS. Besides, the accuracy of the RIS-enabled FL system is close to the case without DP. Furthermore, we can observe that our proposed methods can achieve good performance even for nonconvex learning tasks.

\begin{figure}
	\begin{minipage}{\linewidth}
		\centerline{\includegraphics[scale=0.65]{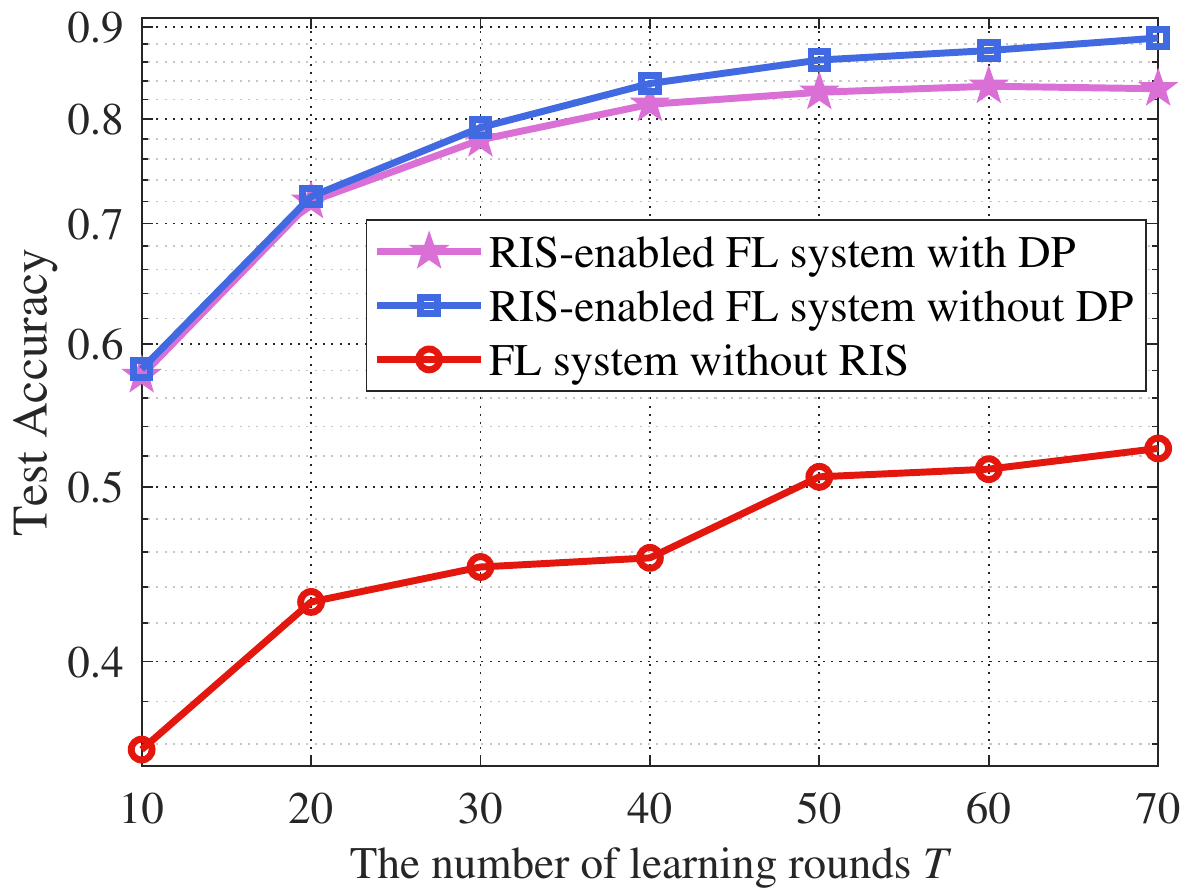}}
		\caption{Test accuracy of the three RIS schemes with MNIST dataset.}
		\label{MNIST}
	\end{minipage}
	\vskip 18pt
	\begin{minipage}{\linewidth}
		\centerline{\includegraphics[scale=0.65]{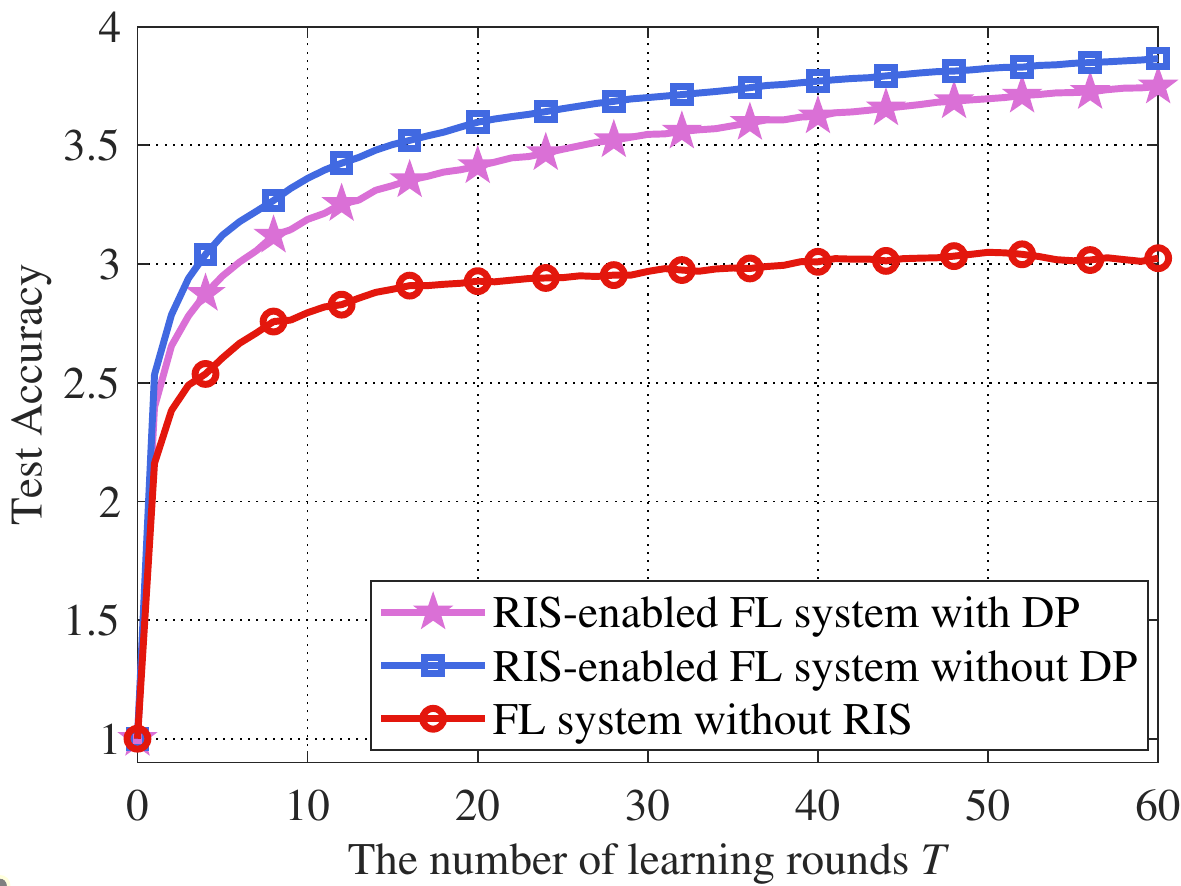}}
		\captionsetup{labelfont={color=black}} 
		\caption{Relative test accuracy of the three RIS schemes with CIFAR-10 dataset.}
		\label{CIFAR10}
	\end{minipage}
\end{figure}

\subsection{Differentially Private FL for Wearable IoT-Based Biomedical Monitoring}
Driven by the massive amounts of biomedical data generated from widespread IoT edge devices, affective computing has become a field of interest in biomedical informatics \cite{AC-survey}. Affective computing is a promising technique that recognizes a person's emotional state based on the physiological signals and images, e.g., photoplethysmogram (PPG) based heart activity, to detect the link between physical health and emotional states. To integrate distributed health data analysis and privacy protection, the authors in \cite{Wearable} established a practical IoT-based wearable biomedical monitoring scenario and provided a first attempt on the application of privacy-preserving FL to IoT-based affective computing. Specifically, 32 public school teachers from Turkey participated in the IoT-based stress detection experiment, where each teacher was required to wear wrist-worn wearable IoT devices and experience three different emotional sessions (baseline, lecture, exam). After each session, the heart activity features collected by wearable IoT devices and the mental state report from each teacher were gathered to construct the stress detection dataset. Besides, the collected data is distributed unevenly across $K=26$ private edge devices, which are orchestrated by an edge server to collaboratively train a global stress detection model. Inspired by this, we simulate a practical wireless IoT-based stress detection FL scenario with DP guarantees to test the learning performance of three RIS schemes.

\par
We train a neural network on the stress detection dataset whose model parameter $\pmb\theta$ with dimension $d=12002$ consists of the the first-layer parameter $\pmb W_1\in\mathbb{R}^{100\times 17}$, the second-layer parameter $\pmb W_2\in\mathbb{R}^{100\times 101}$ and the third-layer parameter $\pmb W_3\in\mathbb{R}^{2\times 101}$. The classification accuracy on the stress detection dataset under different levels of privacy and SNR is presented in Table \ref{table3}. It can be observed that the RIS-enabled FL system with DP is capable of outperforming the FL system without RIS in terms of classification accuracy and robustness to wireless noise. The RIS-enabled system with DP can also achieve similar performance compared with the system without DP.

\section{Conclusions}\label{VI}
In this paper, we developed an RIS-enable differentially private FL system by leveraging the reconfigurability of channel propagation via RIS and the property of waveform superposition via AirComp. We theoretically characterized the convergence behavior of the over-the-air FL algorithm, for which a system optimization problem was established to achieve better learning accuracy under the privacy and transmit power constraints. We further proposed a two-step low-rank optimization framework to minimize the learning optimality gap, by jointly optimizing the power allocation, artificial noise, and reflecting coefficients of RIS during the learning procedure. Through convergence analysis and system optimization, we revealed that the RIS-enabled FL system is able to achieve higher system SNR and boost the receive signal power,  thereby improving learning accuracy performance while satisfying the privacy requirements. Numerical results also demonstrated that the proposed FL system can achieve higher learning accuracy and privacy than the benchmarks.

\appendices
\section{Proof of Theorem 1} \label{Proof-th1}
For each edge device $k$, we note that the local gradient $\pmb g_{k,t}$ is divided into $I$ $e$-dimensional sub-signals. Based on (\ref{Estimated-g-entry}), we arrive at the equivalent channel noise vector
\begin{equation}
\hat{\pmb w}_t=\left[\frac{\pmb w_t^T(1)}{\sqrt{\eta_t(1)}},\ldots,\frac{\pmb w_t^T(I)}{\sqrt{\eta_t(I)}}\right]^T.
\end{equation}
Then the estimated $\hat{\pmb g}_t$ at the edge server can be expressed as
\begin{equation}
\hat{\pmb g}_t=\pmb g_t+\frac{1}{KD}\sum_{k=1}^{K}\text{Re}\{\pmb n_{k,t}\}+\frac{1}{KD}\text{Re}\left\{\hat{\pmb w}_t\right\}.
\end{equation}
Hence, under Assumption 2, we have the following inequality
\begin{align}
F(&\pmb\theta_{t+1})\leq F\left(\pmb\theta_{t}\right)-\lambda\pmb g_t^T\hat{\pmb g}_{t} +\frac{L\lambda^2}{2}\left|\left|\hat{\pmb g}_{t}\right|\right|^2 \nonumber \\
&=F\left(\pmb\theta_{t}\right)-\lambda\left|\left|\pmb g_t\right|\right|^2-\frac{\lambda}{KD}\sum_{k=1}^{K}\,\text{Re}\left\{\pmb g_t^T\pmb n_{k,t}\right\}-\frac{\lambda\text{Re}\left\{\pmb g_t^T\hat{\pmb w}_t\right\}}{KD} \nonumber \\
&\qquad\qquad+\frac{L\lambda^2}{2}\Bigg|\Bigg| \pmb g_t+\frac{1}{KD}\sum_{k=1}^{K}\text{Re} \left\{\pmb n_{k,t}\right\}+\frac{1}{KD}\text{Re}\left\{\hat{\pmb w}_t\right\}\Bigg|\Bigg|^2. \nonumber
\end{align}

\par
Recalling that $\pmb n_{k,t}\!=\![\pmb n_{k,t}^T(1),\ldots,\pmb n_{k,t}^T(I)]^T$ with $\pmb n_{k,t}(i)\!\sim\!\mathcal{CN}(\pmb 0,\sigma_{k,t}^2(i)\,\pmb I_e)$ and given $\lambda\!=\!1/L$, by taking the expectations over the additive noise (including the artificial and wireless channel noise) on both sides of the above inequality, we obtain
\begin{equation}
\begin{split}
&\mathbb{E}[F(\pmb\theta_{t+1})]\leq F\left(\pmb\theta_{t}\right)+\left(\frac{L}{2}\lambda^2-\lambda\right)\left|\left|\pmb g_t\right|\right|^2+\frac{L\lambda^2}{2}\times \nonumber \\
&\qquad\frac{1}{(KD)^2}\mathbb{E}\Bigg[\Big|\Big|\sum_{k=1}^{K}\text{Re}\left\{[\pmb n_{k,t}^T(1),\ldots,\pmb n_{k,t}^T(I)]^T\right\}+\text{Re}\left\{\hat{\pmb w}_t\right\}\Big|\Big|^2\Bigg] \\
&\quad\leq F\left(\pmb\theta_{t}\right)\!-\!\frac{1}{2L}\left|\left|\pmb g_t\right|\right|^2\!+\!\frac{e}{4L(KD)^2}\left(\sum_{k=1}^{K}\sum_{i=1}^{I}\sigma_{k,t}^2(i)\!+\!\sum_{i=1}^{I}\frac{N_0}{\eta_t(i)}\right). \\
\end{split}
\end{equation}

\par
According to (\ref{PL}), by setting $\pmb\theta=\pmb\theta^*$ and $\pmb\theta'=\pmb\theta_{t}$, we have
\begin{equation}
\frac{1}{2}\left|\left|\nabla F\left(\pmb\theta_{t}\right)\right|\right|^2\geq\mu\left[F\left(\pmb\theta_{t}\right)-F^*\right].
\end{equation}
Therefore, subtracting the optimal value $F^*$ at both sides yields
\begin{equation}
\begin{split}
\mathbb{E}\left[F\left(\pmb\theta_{t+1}\right)\right]&-F^*\leq\left(1-\frac{\mu}{L}\right)\left[F\left(\pmb\theta_{t}\right)-F^*\right]+ \\
&\qquad\frac{e}{4L(KD)^2}\left(\sum_{k=1}^{K}\sum_{i=1}^{I}\sigma_{k,t}^2(i)\!+\!\sum_{i=1}^{I}\frac{N_0}{\eta_t(i)}\right). \label{Iteration-proof}
\end{split}
\end{equation}

\par
Finally, the expected result is obtained by applying (\ref{Iteration-proof}) iteratively through $T$ learning rounds.  \qed

\section{Proof of Theorem 2} \label{Proof-th2}
The proving process is based on the advanced literature \cite{FLPrivacy-free,Differential-privacy1,DP-proof}. We focus on the privacy constraint of the $k$-th edge device based on the aggregated signals $\pmb r=\{\pmb r_t\}_{t=1}^{T}$ at the edge server. It is worth noting that, without traditional lightweight assumption, our transmission model faces a more practical scenario where one learning round consists of $I$ communication blocks.

\par
According to (\ref{Privacy-loss}), the privacy loss after $T$ learning rounds can be expressed as
\begin{equation}
\begin{split}
\mathcal{L}_{\mathcal{D}_k,\mathcal{D}_k'}\left(\pmb r\right)&=\ln\left\{\prod_{t=1}^T\frac{\Pr\left[\pmb r_{t}\big|\pmb r_{t-1},\ldots,\pmb r_{1};\mathcal{D}_k\right]}{\Pr\left[\pmb r_{t}\big|\pmb r_{t-1},\ldots,\pmb r_{1};\mathcal{D}_k'\right]}\right\}  \\
&=\sum_{t=1}^{T}\ln\left\{\frac{\Pr\left[\pmb r_{t}\big|\pmb r_{t-1},\ldots,\pmb r_{1};\mathcal{D}_k\right]}{\Pr\left[\pmb r_{t}\big|\pmb r_{t-1},\ldots,\pmb r_{1};\mathcal{D}_k'\right]}\right\}. \label{Privacy-loss-proof1}
\end{split}
\end{equation}

\par
The effective noise $\pmb q_t$ in (\ref{Effective-noise}) is a complex Gaussian random vector with statistically independent elements. Besides, recalling that the noise vector $\pmb w_{t}(i)\!\sim\!\mathcal{CN}(\pmb 0,N_0\pmb I_e)$ and $\pmb n_{k,t}(i)\!\sim\!\mathcal{CN}(\pmb 0,\sigma_{k,t}^2(i)\,\pmb I_e)$, we obtain the pseudo-covariance matrix of $\pmb q_t$, which is given by
\begin{equation}
\pmb J_t=\mathbb{E}\left[\pmb q_t\left(\pmb q_t\right)^T\right]=\pmb 0, \label{Pseudo-covariance-matrix-proof}
\end{equation}
and the covariance matrix of $\pmb q_t$ can be expressed as
\begin{equation}
\pmb\Sigma_t=\mathbb{E}\left[\pmb q_t\left(\pmb q_t\right)^H\right]=
\begin{bmatrix}
\pmb\Lambda_t(1) & \cdots & \pmb 0  \\
\vdots & \ddots & \vdots  \\
\pmb 0 & \cdots & \pmb\Lambda_t(I)  \\
\end{bmatrix}, \label{Covariance-matrix}
\end{equation}
where 
\begin{equation}
\pmb\Lambda_t(i)=\left(\eta_t(i)\sum_{k=1}^{K}\sigma_{k,t}^2(i)+N_0\right)\pmb I_e,\,\forall i\in[I]
\end{equation}
denotes the diagonal positive semi-definite covariance matrix of $\pmb q_t(i)$, which indicates that $\pmb\Sigma_t$ is a diagonal positive semi-definite matrix.

\par
According to \cite[Appendix. A]{Circular-symmetric-proof-2}, it is easy to verify that $\pmb q_t$ is a complex random vector with circular symmetry property, whose probability density function (pdf) is given by
\begin{equation}
p(\pmb q_t)=\frac{1}{\pi^{d}\text{det}(\pmb\Sigma_t)}\exp\left(-\pmb q_t^H\pmb\Sigma_t^{-1}\pmb q_t\right). \label{Pdf-of-Rkt-proof}
\end{equation}
Therefore, for the analysis of sensitivity, we further denote the effective noise based on dataset $\mathcal{D}_k$ as $\pmb z_{k,t}=[\pmb z_{k,t}^T(1),\ldots,\pmb z_{k,t}^T(I)]^T$ with
\begin{equation}
\pmb z_{k,t}(i)=\pmb r_{t}(i)-\sqrt{\eta_t(i)}D_k\pmb g_{k,t}(i;\mathcal{D}_k)-\nu_{k,t}(i), \nonumber
\end{equation} 
where $\nu_{k,t}(i)$ denotes the second constant term in (\ref{Rt}) and $\pmb g_{k,t}(i;\mathcal{D}_k)$ represents the updated gradient based on $\mathcal{D}_k$, thus we can obtain
\begin{equation}
\Pr[\pmb r_{t}\big|\pmb r_{t-1},\ldots,\pmb r_{1};\mathcal{D}_k]=\frac{1}{\pi^{d}\text{det}(\pmb\Sigma_t)}\exp\left(-\pmb z_{k,t}^H\pmb\Sigma_t^{-1}\pmb z_{k,t}\right). \nonumber
\end{equation}
Furthermore, according to (\ref{Ukt}), $\pmb z_{k,t}+\pmb u_{k,t}$ represents the effective noise based on dataset $\mathcal{D}_k'$. Hence, substituting them into (\ref{Privacy-loss-proof1}), we have
\begin{equation}
\begin{split}
&\mathcal{L}_{\mathcal{D}_k,\mathcal{D}_k'}\left(\pmb r\right) \\
&\quad=\sum_{t=1}^{T}\left[\left(\pmb z_{k,t}+\pmb u_{k,t}\right)^H\pmb\Sigma_t^{-1}\left(\pmb z_{k,t}+\pmb u_{k,t}\right)-\left(\pmb z_{k,t}\right)^H\pmb\Sigma_t^{-1}\pmb z_{k,t}\right] \\
&\quad=\sum_{t=1}^{T}\left[2\text{Re}\left\{\pmb u_{k,t}^H\pmb\Sigma_t^{-1}\pmb z_{k,t}\right\}+\pmb u_{k,t}^H\pmb\Sigma_t^{-1}\pmb u_{k,t}\right]. \nonumber
\end{split}
\end{equation}
Based on (\ref{Privacy-standard}), given the DP parameter pair $(\epsilon,\delta)$, the privacy violation probability bound can be expressed as
\begin{equation}
\begin{split}
\Pr&\left(\left|\mathcal{L}_{\mathcal{D}_k,\mathcal{D}_k'}(\pmb r)\right|>\epsilon\right) \\
&\overset{\text{(a)}}{\leq}\Pr\left[\left|\text{Re}\left\{\pmb u_{k,t}^H\pmb\Sigma_t^{-1}\pmb z_{k,t}\right\}\right|>\frac{\epsilon}{2}-\frac{1}{2}\sum_{t=1}^{T}\pmb u_{k,t}^H\pmb\Sigma_t^{-1}\pmb u_{k,t}\right] \\
&=2\Pr\left[\sum_{t=1}^{T}\text{Re}\left\{\pmb u_{k,t}^H\pmb\Sigma_t^{-1}\pmb z_{k,t}\right\}>\frac{\epsilon}{2}-\frac{1}{2}\sum_{t=1}^{T}\pmb u_{k,t}^H\pmb\Sigma_t^{-1}\pmb u_{k,t}\right] \\
&\triangleq 2\Pr\left(\Upsilon_k>c_k\right),\nonumber
\end{split}
\end{equation}
where (a) comes from the inequality $\Pr(|X+e|>\epsilon)\leq\Pr(|X|>\epsilon-e)$ for $e\geq 0$. We note that $\pmb z_{k,t}\sim\mathcal{CN}(\pmb 0,\pmb\Sigma_t)$ and $\pmb\Sigma_t$ is a diagonal matrix. Hence, it is easy to verify that $\pmb\Sigma_t^{-1/2}\pmb z_{k,t}\sim\mathcal{CN}(\pmb 0,\pmb I_{d})$, which leads to $\pmb u_{k,t}^H\pmb\Sigma_t^{-1}\pmb z_{k,t}\sim\mathcal{CN}(\pmb 0,\pmb u_{k,t}^H\pmb\Sigma_t^{-1}\pmb u_{k,t})$. As a result, considering the fact that any $\pmb z_{k,t}$ is statistical independent of each other, $\Upsilon_k$ is a random variable according to $\mathcal{N}(0,\frac{1}{2}\sum_{t=1}^{T}\pmb u_{k,t}^H\pmb\Sigma_t^{-1}\pmb u_{k,t})$.

\par
According to (\ref{Delta-bound}) and by denoting
\begin{equation}
\xi_t\leq\eta_t(i)\sum_{k=1}^{K}\sigma_{k,t}^2(i)+N_0,\, \forall i,t \nonumber
\end{equation}
as a lower bound of the power of effective noise, we have the following inequalities 
\begin{equation}
\begin{split}
v_k^2&\triangleq\sum_{t=1}^{T}\frac{2\gamma_t^2}{\xi_t}\underset{i\in[I]}{\max}\left\{\eta_t(i)\right\}\overset{(a)}{\geq}\sum_{t=1}^{T}\frac{||\pmb u_{k,t}||^2}{2\xi_t}\overset{(b)}{\geq}\frac{1}{2}\sum_{t=1}^{T}\pmb u_{k,t}^H\pmb\Sigma_t^{-1}\pmb u_{k,t}, \nonumber \\
&\;\hat{c}_k\leq\frac{\epsilon}{2}-\sum_{t=1}^{T}\frac{2\gamma_t^2}{\xi_t}\underset{i\in[I]}{\max}\left\{\eta_t(i)\right\}\overset{(c)}{\leq}\frac{\epsilon}{2}-\sum_{t=1}^{T}\frac{||\pmb u_{k,t}||^2}{2\xi_t}\leq c_k,
\end{split}
\end{equation}
where (a) and (c) comes from the upper bound in (\ref{Delta-bound}), and (b) is obtained by (\ref{Covariance-matrix}). Based on this, we introduce another random variable $\Xi_k$ according to $\mathcal{CN}(0,v_k^2)$. Thus, we arrive at the following inequalities:
\begin{equation}
\Pr(\Upsilon_k>c_k)\leq\Pr(\Xi_k>c_k)\leq\Pr(\Xi_k>\hat{c}_k).
\end{equation}

\par
Finally, according to Mill's inequality, we can obtain
\begin{equation}
\Pr\left(|\Upsilon_k|>c_k\right)<\sqrt{\frac{2}{\pi}}\frac{v_k}{\hat{c}_k}\exp\left(-\frac{\hat{c}_k^2}{2v_k^2}\right)<\delta, \nonumber
\end{equation}
followed by solving an equation based on the monotonicity property of function $\mathcal{C}(x)=\sqrt{\pi}xe^{x^2}$, then we arrive at the expected result in (\ref{Privacy-constraints}).\qed

\section{Proof of Theorem 3} \label{Proof-th3}
To start, by using the change of variables, we define the following new relationships:
\begin{equation}
a_{k,t}^i=\sigma_{k,t}^2(i),\, b_t^i=\eta_t^{-1}(i),\, c_t^i=\xi_t b_t^i=\xi_t\eta_t^{-1}(i). \label{Variable-change1}
\end{equation}
Thus, it is easy to verify that $a_{k,t}^i\geq 0$ and $b_t^i,c_t^i>0$. Besides, we also denote
\begin{equation}
d_t=\frac{\xi_t}{\max_i\{\eta_t(i)\}}=\xi_t\min_i\{b_t^i\},\,\forall t, \label{Variable-change2}
\end{equation}
which yields an inherent constraint $d_t\leq c_t^i,\forall i,t$.

\par
For short, we omit the superscript and subscript when describing the optimization variables, e.g., $a_{k,t}^i$ as $a$. Therefore, problem $\mathscr{P}_{1}$ can be transformed into problem $\mathscr{P}_{1.1}$:
\begin{subequations}
        \begin{align}
        & \;\,
        \underset{\{a,b,c,d\}}{\text{min}}
        & &
        \sum_{t=1}^{T}\left(1-\frac{\mu}{L}\right)^{-t}\left(\sum_{k=1}^{K}\sum_{i=1}^{I}a_{k,t}^i+N_0\sum_{i=1}^{I}b_{t}^i\right) \label{P1.1-obj-proof} \\
        & \text{subject to}
        & &
        \sum_{t=1}^{T}\frac{4\gamma_t^2}{d_t}-\mathcal{R}_{dp}\leq 0,  \label{P1.1-privacy-proof} \\
        &&&
        c_t^i-\sum_{k=1}^{K}a_{k,t}^i-N_0 b_t^i\leq 0,\,\forall i,t, \label{P1.1-cti-proof} \\
        &&&
        D^2\zeta_{k,t}^2(i)+ea_{k,t}^i\leq P_0|h_{k,t}(i)|^2b_t^i,\,\forall i,k,t, \label{P1.1-power-proof} \\
        &&&
        d_t-c_t^i\leq 0,\,\forall i,t \label{P1.1-accessional-proof}, \\
        &&&
        a_{k,t}^i\geq 0,\,b_t^i\geq 0,\,c_t^i\geq 0,\,d_t\geq 0,\,\forall i,k,t, \label{P1.1-inherent-proof}
        \end{align}
\end{subequations}
which is a pure convex problem and can be tackled by KKT conditions. Above all, the Lagrange function is given as
\begin{align*}
\mathscr{L}&=\sum_{t=1}^{T}\left(1-\frac{\mu}{L}\right)^{-t}\left(\sum_{k=1}^{K}\sum_{i=1}^{I}a_{k,t}^i+N_0\sum_{i=1}^{I}b_{t}^i\right) \\
&\quad+\beta\,\left(\sum_{t=1}^{T}\frac{4\gamma_t^2}{d_t}-\mathcal{R}_{dp}\right)-\sum_{t=1}^{T}\sum_{k=1}^{K}\sum_{i=1}^{I}\varphi_{k,t}^i a_{k,t}^i \\
&\quad+\sum_{t=1}^{T} \sum_{i=1}^{I}\iota_{t}^i\left(c_t^i-\sum_{k=1}^{K}a_{k,t}^i-N_0b_t^i\right)+\sum_{t=1}^{T}\sum_{i=1}^{I}u_t^i\left(d_t-c_t^i\right) \\
&\quad+\sum_{t=1}^{T}\sum_{k=1}^{K}\sum_{i=1}^{I}\psi_{k,t}^i\left(D^2\zeta_{k,t}^2(i)+ea_{k,t}^i-P_0|h_{k,t}(i)|^2b_t^i\right), \nonumber
\end{align*} 
where $\beta,\iota_t^i,u_k^i,\psi_{k,t}^i,\varphi_{k,t}^i\geq 0$ represent the Lagrange multipliers and for simplicity, we omit some obviously zero terms. We denote variables with "$\sim$" as the optimal solutions satisfying KKT conditions, e.g., $\widetilde{a}_{k,t}^i$ denotes the optimal $a_{k,t}^i$. Hence, we can obtain the following relationships:
\begin{subequations}
	\begin{align}
        &\frac{\partial\mathscr{L}}{\partial a_{k,t}^i}=\left(1-\frac{\mu}{L}\right)^{-t}-\iota_{t}^i+e\psi_{k,t}^i-\varphi_{k,t}^i=0, \label{KKT-akt} \\
        &\frac{\partial\mathscr{L}}{\partial b_{t}^i}=N_0\left(1-\frac{\mu}{L}\right)^{-t}-N_0\iota_{t}^i-P_0\sum_{k=1}^{K}\psi_{k,t}^i|h_{k,t}(i)|^2=0, \label{KKT-bti} \\
        &\frac{\partial\mathscr{L}}{\partial c_t^i}=\iota_t^i-u_t^i=0, \label{KKT-cti} \\
        &\frac{\partial\mathscr{L}}{\partial d_t}=-\beta\frac{4\gamma_t^2}{d_t^2}+\sum_{i=1}^{I}u_t^i=0, \label{KKT-dt}
	\end{align}
\end{subequations}
which indicate that at least one of $\widetilde{\iota}_t^i$ and $\widetilde{\varphi}_{k,t}^i$ is greater than $0$. Based on this, we mainly focus on the following three conditions and elaborate their optimality.

\par
\emph{\pmb{Power field}}: On the one hand, we first focus on a simple case where the inequality in (\ref{P1.1-privacy-proof}) strictly holds, i.e., the privacy constraint is not dominant. Hence, according to complementary slackness condition, we immediately get 
\begin{equation}
\widetilde{\beta}=\widetilde{u}_t^i=\widetilde{\iota}_t^i=0,\,\forall i,t.
\end{equation}
Furthermore, we have $\widetilde{\varphi}_{k,t}>0$ based on (\ref{KKT-akt}) which also implies $\widetilde{a}_{k,t}^i=0$. Hence, the power constraint (\ref{P1.1-power-proof}) becomes the unique one that restricts the learning accuracy. Based on this, we arrive at the unique optimal solution in this case, i.e.,
\begin{equation}
\widetilde{b}_t^i=\frac{1}{P_0}\underset{k\in[K]}{\max}\frac{D^2\zeta_{k,t}^2(i)}{\;|h_{k,t}(i)|^2}.
\end{equation}

\par
On the other hand, we turn to the difficult multi-solution case where the privacy constraint (\ref{P1.1-privacy-proof}) is stringent with $\widetilde{\beta}\neq 0$, which results in the following equality:
\begin{equation}
\sum_{t=1}^{T}\frac{4\gamma_t^2}{d_t}-\mathcal{R}_{dp}=0. \label{Privacy-equ-proof}
\end{equation}
Motivated by the previous condition, we set $\widetilde{a}_{k,t}^i=0$, which saves the communication and power resources of the edge server. Hence, we arrive at the following two conditions.

\par
\emph{\pmb{Privacy field}}: We first focus on the case where only the privacy constraint takes effect, which indicates (\ref{P1.1-power-proof}) strictly holds, i.e.,
\begin{equation}
\widetilde{b}_t>\frac{1}{P_0}\underset{i,k}{\max} \frac{D^2\zeta_{k,t}^2(i)}{\;|h_{k,t}(i)|^2}, \label{Bt-lower}
\end{equation}
yielding $\widetilde{\psi}_{k,t}^i\!=\!0$. Hence, based on (\ref{KKT-bti}) and (\ref{KKT-cti}), we obtain
\begin{equation}
\widetilde{u}_t^i=\widetilde{\iota}_t^i=(1-\mu/L)^{-t},\,\forall i,t. \label{Proof-uti}
\end{equation}
Besides, combining (\ref{KKT-dt}) and complementary slackness condition, we further get
\begin{equation}
\widetilde{d}_t=\widetilde{c}_t^i=N_0\widetilde{b}_t^i=\frac{2\gamma_t}{\sqrt{I}}(\widetilde{\beta})^{\frac{1}{2}}\left(1-\frac{\mu}{L}\right)^{\frac{t}{2}}, \label{KKT-sol1}
\end{equation}
which indicates $\widetilde{b}_t^i$ remains constant during one learning round, and for short, we denote $\widetilde{b}_t=\widetilde{b}_t^i,\forall i$. In summary, the optimal conditions of $\widetilde{b}_t$ are given by (\ref{Privacy-equ-proof}), (\ref{Bt-lower}), and (\ref{KKT-sol1}).

\par
\emph{\pmb{Privacy-power field}}: Now, we turn to analyze the case where the dual constraints of privacy and power take effect, i.e., there exist some $\widetilde{c}_t^i\neq\widetilde{d}_t$ which lead to $\widetilde{u}_t^i=\widetilde{\iota}_t^i=0$. Besides, from (\ref{KKT-bti}), we can verify that there exist $\widetilde{\psi}_{k,t}^i\neq 0$ based on the power constraints of all edge devices. Remarkably, $\widetilde{b}_t^i$ remains constant across all edge devices in the $i$-th communication block of learning round $t$, so it must meet all power constraints and strictly satisfy a certain one, i.e.,
\begin{equation}
\widetilde{b}_t^i=\frac{1}{P_0}\underset{k}{\max}\frac{D^2\zeta_{k,t}^2(i)}{\;|h_{k,t}(i)|^2}.
\end{equation}

\par
As for the other $\widetilde{u}_t^i\neq 0$, we consider it as a generalized version of (\ref{Proof-uti}), i.e., some communication blocks are restricted by privacy while the others are power. For analytical ease, we introduce parameter $\tau_t$ to represent the number of communication blocks limited by privacy in the $t$-th learning round, i.e., $\widetilde{u}_{k,t}^i\!\neq\! 0$. Hence, (\ref{KKT-dt}) can be reformulated as
\begin{equation}
\tau_t\left(1-\frac{\mu}{L}\right)^{-t}=\frac{4\gamma_t^2}{\widetilde{d}_t^2}\widetilde{\beta}.
\end{equation}
For instance, if all communication blocks are restricted by privacy, $\tau_t$ is set to $I$ which is the same as (\ref{KKT-sol1}), i.e., \emph{\pmb{Privacy}} \emph{\pmb{field}}. Therefore, considering the above two cases, we arrive at the expected result by elaborately selecting $\tau_t$ in each learning round $t$ and parameter $\widetilde{\beta}$ to strictly meet (\ref{Privacy-equ-proof}).

\par
Finally, according to (\ref{Variable-change1}), we obtain the expected result. \qed

\bibliographystyle{ieeetr}
\bibliography{refs}

\begin{thebibliography}{10}

\bibitem{Intorduction-intelligent-networks}
K.~B. Letaief, W.~Chen, Y.~Shi, J.~Zhang, and Y.-J.~A. Zhang, ``The roadmap to
  6{G}: {AI} empowered wireless networks,'' {\em IEEE Commun. Mag.}, vol.~57,
  no.~8, pp.~84--90, 2019.

\bibitem{GoogleFL}
B.~McMahan, E.~Moore, D.~Ramage, S.~Hampson, and B.~A. Y~Arcas,
  ``Communication-efficient learning of deep networks from decentralized
  data,'' in {\em Proc. Int. Conf. Artif. Intell. Stat. (AISTATS)}, vol.~54,
  pp.~1273--1282, PMLR, 2017.

\bibitem{6G-introduction}
K.~B. Letaief, Y.~Shi, J.~Lu, and J.~Lu, ``Edge artificial intelligence for
  6{G}: Vision, enabling technologies, and applications,'' {\em IEEE J. Sel.
  Areas Commun.}, vol.~40, no.~1, pp.~5--36, 2022.

\bibitem{FL-challenges2}
Y.~Shi, K.~Yang, T.~Jiang, J.~Zhang, and K.~B. Letaief,
  ``Communication-efficient edge {AI}: Algorithms and systems,'' {\em IEEE
  Commun. Surveys Tuts.}, vol.~22, no.~4, pp.~2167--2191, 2020.

\bibitem{Dinh_CST21}
D.~C. Nguyen, M.~Ding, P.~N. Pathirana, A.~Seneviratne, J.~Li, and
  H.~Vincent~Poor, ``Federated learning for internet of things: A comprehensive
  survey,'' {\em IEEE Commun. Surveys Tuts.}, vol.~23, no.~3, pp.~1622--1658,
  2021.

\bibitem{Ahmed_IoTJ21}
A.~Imteaj, U.~Thakker, S.~Wang, J.~Li, and M.~H. Amini, ``A survey on federated
  learning for resource-constrained {IoT} devices,'' {\em IEEE Internet of
  Things J.}, 2021, Doi: 10.1109/JIOT.2021.3095077.

\bibitem{wang2021field}
J.~Wang, Z.~Charles, Z.~Xu, G.~Joshi, H.~B. McMahan, M.~Al-Shedivat, G.~Andrew,
  S.~Avestimehr, K.~Daly, D.~Data, {\em et~al.}, ``A field guide to federated
  optimization,'' {\em arXiv:2107.06917. [Online].
  https://arxiv.org/pdf/2107.06917.pdf}, 2021.

\bibitem{rieke2020future}
N.~Rieke, J.~Hancox, W.~Li, F.~Milletari, H.~R. Roth, S.~Albarqouni, S.~Bakas,
  M.~N. Galtier, B.~A. Landman, K.~Maier-Hein, {\em et~al.}, ``The future of
  digital health with federated learning,'' {\em NPJ digital medicine}, vol.~3,
  no.~1, pp.~1--7, 2020.

\bibitem{IoMT1}
H.~Jin, X.~Dai, J.~Xiao, B.~Li, H.~Li, and Y.~Zhang, ``Cross-cluster federated
  learning and blockchain for internet of medical things,'' {\em IEEE Internet
  Things J.}, 2021.

\bibitem{IoMT2}
V.~Hayyolalam, M.~Aloqaily, O.~Ozkasap, and M.~Guizani, ``Edge intelligence for
  empowering iot-based healthcare systems,'' {\em IEEE Wireless
  Communications}, vol.~28, no.~3, pp.~6--14, 2021.

\bibitem{Yan_JBHI21}
Z.~Yan, J.~Wicaksana, Z.~Wang, X.~Yang, and K.-T. Cheng, ``Variation-aware
  federated learning with multi-source decentralized medical image data,'' {\em
  IEEE J. Biomed. Health Inform.}, vol.~25, no.~7, pp.~2615--2628, 2021.

\bibitem{chen2020fl}
S.~Chen, D.~Xue, G.~Chuai, Q.~Yang, and Q.~Liu, ``Fl-{QSAR}: a federated
  learning-based qsar prototype for collaborative drug discovery,'' {\em
  Bioinformatics}, vol.~36, no.~22-23, pp.~5492--5498, 2020.

\bibitem{warnat2021swarm}
S.~Warnat-Herresthal, H.~Schultze, K.~L. Shastry, S.~Manamohan, S.~Mukherjee,
  V.~Garg, R.~Sarveswara, K.~H{\"a}ndler, P.~Pickkers, N.~A. Aziz, {\em
  et~al.}, ``Swarm learning for decentralized and confidential clinical machine
  learning,'' {\em Nature}, vol.~594, no.~7862, pp.~265--270, 2021.

\bibitem{FL-challenges1}
T.~Li, A.~K. Sahu, A.~Talwalkar, and V.~Smith, ``Federated learning:
  Challenges, methods, and future directions,'' {\em IEEE Signal Process.
  Mag.}, vol.~37, no.~3, pp.~50--60, 2020.

\bibitem{Privacy-introduction}
U.~Sennur, A.~Salman, G.~Michael, J.~Syed, T.~Ravi, and T.~Chao, ``Private
  retrieval, computing and learning: Recent progress and future challenges,''
  {\em arXiv:2108.00026. [Online]. https://arxiv.org/pdf/ 2108.00026.pdf},
  2021.

\bibitem{AirComp}
K.~Yang, T.~Jiang, Y.~Shi, and Z.~Ding, ``Federated learning via over-the-air
  computation,'' {\em IEEE Trans. Wireless Commun.}, vol.~19, no.~3,
  pp.~2022--2035, 2020.

\bibitem{Broadband}
G.~Zhu, Y.~Wang, and K.~Huang, ``Broadband analog aggregation for low-latency
  federated edge learning,'' {\em IEEE Trans. Wireless Commun.}, vol.~19,
  no.~1, pp.~491--506, 2019.

\bibitem{AirComp3}
M.~M. Amiri and D.~G{\"u}nd{\"u}z, ``Machine learning at the wireless edge:
  Distributed stochastic gradient descent over-the-air,'' {\em IEEE Trans.
  Signal Process.}, vol.~68, pp.~2155--2169, 2020.

\bibitem{Tao_TWC21}
N.~Zhang and M.~Tao, ``Gradient statistics aware power control for over-the-air
  federated learning,'' {\em IEEE Trans. Wireless Commun.}, vol.~20, no.~8,
  pp.~5115--5128, 2021.

\bibitem{Introduction-RIS}
K.~Yang, Y.~Shi, Y.~Zhou, Z.~Yang, L.~Fu, and W.~Chen, ``Federated machine
  learning for intelligent {I}o{T} via reconfigurable intelligent surface,''
  {\em IEEE Network}, vol.~34, no.~5, pp.~16--22, 2020.

\bibitem{Channel-coefficient}
Z.~Wang, J.~Qiu, Y.~Zhou, Y.~Shi, L.~Fu, W.~Chen, and K.~B. Letaief,
  ``Federated learning via intelligent reflecting surface,'' {\em IEEE Trans.
  Wireless Commun.}, 2021.

\bibitem{Channel-invariant}
H.~Liu, X.~Yuan, and Y.-J.~A. Zhang, ``Reconfigurable intelligent surface
  enabled federated learning: A unified communication-learning design
  approach,'' {\em IEEE Trans. Wireless Commun.}, 2021.
\newblock doi:{10.1109/TWC.2021.3086116}.

\bibitem{RIS-challenges}
X.~Yuan, Y.-J.~A. Zhang, Y.~Shi, W.~Yan, and H.~Liu,
  ``Reconfigurable-intelligent-surface empowered wireless communications:
  Challenges and opportunities,'' {\em IEEE Wireless Commun.}, vol.~28, no.~2,
  pp.~136--143, 2021.

\bibitem{RIS-principle}
Q.~Wu and R.~Zhang, ``Intelligent reflecting surface enhanced wireless network
  via joint active and passive beamforming,'' {\em IEEE Trans. Wireless
  Commun.}, vol.~18, no.~11, pp.~5394--5409, 2019.

\bibitem{Huang_TWC19}
C.~Huang, A.~Zappone, G.~C. Alexandropoulos, M.~Debbah, and C.~Yuen,
  ``Reconfigurable intelligent surfaces for energy efficiency in wireless
  communication,'' {\em IEEE Trans. Wireless Commun.}, vol.~18, no.~8,
  pp.~4157--4170, 2019.

\bibitem{Survey-RIS}
M.~Di~Renzo, A.~Zappone, M.~Debbah, M.-S. Alouini, C.~Yuen, J.~De~Rosny, and
  S.~Tretyakov, ``Smart radio environments empowered by reconfigurable
  intelligent surfaces: How it works, state of research, and the road ahead,''
  {\em IEEE J. Sel. Areas Commun.}, vol.~38, no.~11, pp.~2450--2525, 2020.

\bibitem{Privacy-leakage1}
B.~Hitaj, G.~Ateniese, and F.~Perez-Cruz, ``Deep models under the {GAN}:
  information leakage from collaborative deep learning,'' in {\em Proc. of the
  2017 ACM SIGSAC Conf. Comput. Commun. Secur.}, pp.~603--618, 2017.

\bibitem{Privacy-leakage2}
L.~Zhu and S.~Han, ``Deep leakage from gradients,'' in {\em Federated
  Learning}, pp.~17--31, Springer, 2020.

\bibitem{FLPrivacy-leakage}
J.~Geiping, H.~Bauermeister, H.~Dr{\"o}ge, and M.~Moeller, ``Inverting
  gradients--how easy is it to break privacy in federated learning?,'' {\em
  arXiv:2003.14053. [Online]. https://arxiv.org/pdf/2003.14053.pdf}, 2020.

\bibitem{Differential-privacy1}
C.~Dwork, A.~Roth, {\em et~al.}, ``The algorithmic foundations of differential
  privacy.,'' {\em Found. Trends Theor. Comput. Sci.}, vol.~9, no.~3-4,
  pp.~211--407, 2014.

\bibitem{Differential-privacy2}
A.~M. Girgis, D.~Data, S.~Diggavi, P.~Kairouz, and A.~T. Suresh, ``Shuffled
  model of federated learning: Privacy, communication and accuracy
  trade-offs,'' {\em IEEE J. Sel. Areas Commun. Inf. Theory (JSAIT)}, vol.~2,
  no.~1, pp.~464--478, 2021.

\bibitem{Laplacian}
N.~Wu, F.~Farokhi, D.~Smith, and M.~A. Kaafar, ``The value of collaboration in
  convex machine learning with differential privacy,'' in {\em 2020 IEEE Secur.
  Priv. (SSP)}, pp.~304--317, IEEE, 2020.

\bibitem{Binomial}
N.~Agarwal, A.~T. Suresh, F.~Yu, S.~Kumar, and H.~B. Mcmahan, ``cp{SGD}:
  Communication-efficient and differentially-private distributed {SGD},'' {\em
  Adv. Neural Inf. Process. Syst.}, pp.~7564--7575, 2018.

\bibitem{FLPrivacy-free}
D.~Liu and O.~Simeone, ``Privacy for free: Wireless federated learning via
  uncoded transmission with adaptive power control,'' {\em IEEE J. Sel. Areas
  Commun.}, vol.~39, no.~1, pp.~170--185, 2020.

\bibitem{Privacy-anonymous}
B.~Has{\i}rc{\i}o{\u{g}}lu and D.~G{\"u}nd{\"u}z, ``Private wireless federated
  learning with anonymous over-the-air computation,'' in {\em 2021 IEEE Int.
  Conf. Acoust. Speech Signal Process. (ICASSP)}, pp.~5195--5199, IEEE, 2021.

\bibitem{Parital-noise}
M.~Seif, R.~Tandon, and M.~Li, ``Wireless federated learning with local
  differential privacy,'' in {\em 2020 IEEE Int. Symp. on Inf. Theory (ISIT)},
  pp.~2604--2609, IEEE, 2020.

\bibitem{Privacy-power-allocation1}
Y.~Koda, K.~Yamamoto, T.~Nishio, and M.~Morikura, ``Differentially private
  aircomp federated learning with power adaptation harnessing receiver noise,''
  in {\em GLOBECOM 2020-2020 IEEE Glob. Commun. Conf.}, pp.~1--6, IEEE, 2020.

\bibitem{Privacy-power-allocation2}
M.~S.~E. Mohamed, W.-T. Chang, and R.~Tandon, ``Privacy amplification for
  federated learning via user sampling and wireless aggregation,'' {\em IEEE J.
  Sel. Areas Commun.}, vol.~39, no.~12, pp.~3821--3835, 2021.

\bibitem{Delay}
L.~Li, L.~Yang, X.~Guo, Y.~Shi, H.~Wang, W.~Chen, and K.~B. Letaief, ``Delay
  analysis of wireless federated learning based on saddle point approximation
  and large deviation theory,'' {\em arXiv:2103.16994. [Online].
  https://arxiv.org/pdf/2103.16994.pdf}, 2021.

\bibitem{DP-proof}
S.~Rezaei~Aghdam, E.~Amid, M.~Furdek, {\em et~al.}, ``Privacy-preserving
  wireless federated learning exploiting inherent hardware impairments,'' {\em
  arXiv:2102.10639. [Online]. https://arxiv.org/pdf/2102.10639.pdf}, 2021.

\bibitem{AirComp4}
T.~Sery and K.~Cohen, ``On analog gradient descent learning over multiple
  access fading channels,'' {\em IEEE Trans. Signal Process.}, vol.~68,
  pp.~2897--2911, 2020.

\bibitem{FLnoniid1}
X.~Li, K.~Huang, W.~Yang, S.~Wang, and Z.~Zhang, ``On the convergence of
  {F}ed{A}vg on non-iid data,'' {\em Proc. Int. Conf. Learn. Represent.
  (ICLR)}, 2020.

\bibitem{Phase-correction}
M.~M. Amiri and D.~G{\"u}nd{\"u}z, ``Machine learning at the wireless edge:
  Distributed stochastic gradient descent over-the-air,'' {\em IEEE Trans. on
  Signal Process.}, vol.~68, pp.~2155--2169, 2020.

\bibitem{RIS-amplitude-1}
X.~Yu, D.~Xu, Y.~Sun, D.~W.~K. Ng, and R.~Schober, ``Robust and secure wireless
  communications via intelligent reflecting surfaces,'' {\em IEEE J. Sel. Areas
  Commun.}, vol.~38, no.~11, pp.~2637--2652, 2020.

\bibitem{Gradient-clipping}
E.~Bagdasaryan, O.~Poursaeed, and V.~Shmatikov, ``Differential privacy has
  disparate impact on model accuracy,'' {\em Adv. Neural Inf. Process. Syst.},
  vol.~32, pp.~15479--15488, 2019.

\bibitem{Transmit-scalar}
L.~Chen, X.~Qin, and G.~Wei, ``A uniform-forcing transceiver design for
  over-the-air function computation,'' {\em IEEE Wireless Commun. Lett.},
  vol.~7, no.~6, pp.~942--945, 2018.

\bibitem{LTE-timing-advance1}
E.~Dahlman, S.~Parkvall, and J.~Skold, {\em 4G: LTE/LTE-advanced for mobile
  broadband}.
\newblock Academic press, 2013.

\bibitem{LTE-timing-advance2}
A.~Ghosh, J.~Zhang, J.~G. Andrews, and R.~Muhamed, {\em Fundamentals of LTE}.
\newblock Pearson Education, 2010.

\bibitem{SGD-hypothesis1}
M.~P. Friedlander and M.~Schmidt, ``Hybrid deterministic-stochastic methods for
  data fitting,'' {\em SIAM J. Sci. Comput.}, vol.~34, no.~3, pp.~A1380--A1405,
  2012.

\bibitem{Approximating}
A.~M.-C. So, J.~Zhang, and Y.~Ye, ``On approximating complex quadratic
  optimization problems via semidefinite programming relaxations,'' {\em Math.
  Program.}, vol.~110, no.~1, pp.~93--110, 2007.

\bibitem{Matrix-lifting}
Z.-Q. Luo, W.-K. Ma, A.~M.-C. So, Y.~Ye, and S.~Zhang, ``Semidefinite
  relaxation of quadratic optimization problems,'' {\em IEEE Signal Process.
  Mag.}, vol.~27, no.~3, pp.~20--34, 2010.

\bibitem{CVX}
M.~Grant and S.~Boyd, ``{CVX}: Matlab software for disciplined convex
  programming, version 2.1.'' http://cvxr.com/cvx, Mar. 2014.

\bibitem{L-bound1}
D.~Wang and J.~Xu, ``Differentially private empirical risk minimization with
  smooth non-convex loss functions: A non-stationary view,'' in {\em Proc.
  Innov. Appl. Artif. Intell. Conf.}, vol.~33, pp.~1182--1189, 2019.

\bibitem{L-bound2}
R.~Bassily, A.~Smith, and A.~Thakurta, ``Private empirical risk minimization:
  Efficient algorithms and tight error bounds,'' in {\em 2014 IEEE 55th Annu.
  Symp. Found. Compute. Sci.}, pp.~464--473, IEEE, 2014.

\bibitem{Clipping-gradient}
X.~Chen, S.~Z. Wu, and M.~Hong, ``Understanding gradient clipping in private
  {SGD}: A geometric perspective,'' {\em Adv. Neural Inf. Process. Syst.},
  vol.~33, 2020.

\bibitem{Channel-prediction1}
N.~Palleit and T.~Weber, ``Time prediction of non flat fading channels,'' in
  {\em 2011 IEEE Int. Conf. Acoust. Speech Signal Process. (ICASSP)},
  pp.~2752--2755, IEEE, 2011.

\bibitem{Channel-prediction2}
H.~Kim, S.~Kim, H.~Lee, C.~Jang, Y.~Choi, and J.~Choi, ``Massive {MIMO} channel
  prediction: Kalman filtering vs. machine learning,'' {\em IEEE Trans.
  Commun.}, vol.~69, no.~1, pp.~518--528, 2021.

\bibitem{Channel-prediction-Kalman2}
C.~Li, J.~Zhang, S.~Song, and K.~B. Letaief, ``Selective uplink training for
  massive {MIMO} systems,'' in {\em 2016 IEEE Int. Conf. Commun. (ICC)},
  pp.~1--6, IEEE, 2016.

\bibitem{Channel-prediction-Kalman1}
B.~Y. Shikur and T.~Weber, ``Channel prediction using an adaptive kalman
  filter,'' in {\em WSA 2015, 19th Int. ITG Workshop Smart Antennas}, pp.~1--7,
  VDE, 2015.

\bibitem{Channel-prediction-ML1}
C.-K. Wen, W.-T. Shih, and S.~Jin, ``Deep learning for massive {MIMO} {CSI}
  feedback,'' {\em IEEE Wireless Commun. Lett.}, vol.~7, no.~5, pp.~748--751,
  2018.

\bibitem{Channel-prediction-ML2}
P.~Dong, H.~Zhang, G.~Y. Li, N.~NaderiAlizadeh, and I.~S. Gaspar, ``Deep {CNN}
  for wideband mmwave massive {MIMO} channel estimation using frequency
  correlation,'' in {\em 2019 IEEE Int. Conf. Acoust. Speech Signal Process.
  (ICASSP)}, pp.~4529--4533, IEEE, 2019.

\bibitem{Deep-learning-CSI}
Y.~Yang, F.~Gao, C.~Xing, J.~An, and A.~Alkhateeb, ``Deep multimodal learning:
  Merging sensory data for massive {MIMO} channel prediction,'' {\em IEEE J.
  Sel. Areas Commun.}, vol.~39, no.~7, pp.~1885--1898, 2020.

\bibitem{Wearable}
Y.~S. Can and C.~Ersoy, ``Privacy-preserving federated deep learning for
  wearable iot-based biomedical monitoring,'' {\em ACM Trans. Internet Technol.
  (TOIT)}, vol.~21, no.~1, pp.~1--17, 2021.

\bibitem{AC-survey}
S.~Greene, H.~Thapliyal, and A.~Caban-Holt, ``A survey of affective computing
  for stress detection: Evaluating technologies in stress detection for better
  health,'' {\em IEEE Cons. Electron. Mag.}, vol.~5, no.~4, pp.~44--56, 2016.

\bibitem{Circular-symmetric-proof-2}
D.~Tse and P.~Viswanath, {\em Fundamentals of wireless communication}.
\newblock Cambridge university press, 2005.

\end{thebibliography}

\end{document}